\documentclass[conference,letterpaper]{IEEEtran}
\IEEEoverridecommandlockouts

\usepackage[utf8]{inputenc} 
\usepackage[T1]{fontenc}
\usepackage{url}
\usepackage{ifthen}
\usepackage{cite}
\usepackage[cmex10]{amsmath}
\usepackage{amsmath}
\usepackage{amsthm}
\usepackage{amsfonts}
\usepackage{amssymb}
\usepackage{mathtools}
\usepackage{color}
\usepackage{verbatim}
\usepackage[normalem]{ulem}
\usepackage{enumitem}
\usepackage[ruled,vlined]{algorithm2e}
\usepackage{algpseudocode}

\usepackage{floatrow}
\usepackage{tikz}
\usetikzlibrary{matrix,arrows}

\definecolor{purple}{rgb}{0.5, 0.0, 0.5}
\definecolor{dark_green}{rgb}{0.0, 0.5, 0.0}
\definecolor{mygray}{gray}{0.6}

\DeclareMathAlphabet{\mathpzc}{OT1}{pzc}{m}{it}

\newcommand{\white}{\color{white}}

\newcommand{\Cc}{\mathcal{C}}
\newcommand{\ow}{\mathcal{O}}

\newcommand{\K}{\mathcal{K}}
\newcommand{\M}{\mathcal{M}}
\newcommand{\Mb}{\bold{M}}

\newcommand{\Qb}{\bold{Q}}

\newcommand{\Qcal}{\mathcal{Q}}

\newcommand{\Rb}{\bold{R}}
\newcommand{\Rt}{\tilde{R}}
\newcommand{\Otil}{\tilde{O}}
\newcommand{\wb}{\bold{w}}

\newcommand{\Ub}{\bold{U}}
\newcommand{\vb}{\bold{v}}
\newcommand{\Vb}{\bold{V}}
\newcommand{\xb}{\bold{x}}
\newcommand{\Xb}{\bold{X}}

\newcommand{\xbh}{{\hat{\bold{x}}}}

\newcommand{\yb}{\bold{y}}

\newcommand{\ellb}{\bar{\ell}}
\newcommand{\ellg}{\grave{\ell}}

\newcommand{\xit}{\tilde{\xi}}
\newcommand{\Fb}{\bold{F}}

\newcommand{\Gbp}{\bold{G}_\mathrm{p}}

\newcommand{\zb}{\bold{z}}
\newcommand{\Zb}{\bold{Z}}
\newcommand{\Pib}{\bold{\Pi}}
\newcommand{\Pibt}{\tilde{\Pib}}
\newcommand{\Omb}{\bold{\Omega}}
\newcommand{\Ombp}{\Omb_\mathrm{p}}
\newcommand{\D}{\mathcal{D}}

\newcommand{\Scal}{\mathcal{S}}

\newcommand{\Z}{\mathbb{Z}}
\newcommand{\R}{\mathbb{R}}
\newcommand{\C}{\mathbb{C}}

\newcommand{\E}{\mathbb{E}}
\newcommand{\N}{\mathbb{N}}

\newcommand{\gt}{\tilde{g}}

\newcommand{\gh}{\hat{g}}
\newcommand{\Pb}{\bold{P}}
\newcommand{\Sb}{\bold{S}}
\newcommand{\Sbt}{\tilde{\bold{S}}}

\newcommand{\Sbp}{\Sb_\mathrm{p}}
\newcommand{\Ab}{\bold{A}}

\newcommand{\Abt}{{\tilde{\bold{A}}}}

\newcommand{\Abph}{\hat{\Ab}_\mathrm{p}}
\newcommand{\Abg}{{\grave{\bold{A}}}}

\newcommand{\bb}{\bold{b}}

\newcommand{\bbh}{\hat{\bold{b}}}
\newcommand{\bbt}{\tilde{\bold{b}}}

\newcommand{\bbg}{{\grave{\bold{b}}}}
\newcommand{\Bb}{\bold{B}}

\newcommand{\ab}{\bold{a}}

\newcommand{\Cb}{\bold{C}}

\newcommand{\Db}{\bold{D}}

\newcommand{\Eb}{\bold{E}}
\newcommand{\eb}{\bold{e}}
\newcommand{\Ib}{\bold{I}}
\newcommand{\Hb}{\bold{H}}
\newcommand{\Hbh}{\hat{\bold{H}}}
\newcommand{\Hbt}{\tilde{\bold{H}}}
\newcommand{\Ht}{\tilde{H}}
\newcommand{\Ubt}{\tilde{\bold{U}}}
\newcommand{\Vbh}{\hat{\bold{V}}}
\newcommand{\Vbt}{\tilde{\bold{V}}}
\newcommand{\Sigb}{\bold{\Sigma}}

\newcommand{\rpm}{\raisebox{.2ex}{$\scriptstyle\pm$}}

\newcommand{\sfsty}[1]{\ensuremath{\mathsf{#1}}}  
\newcommand{\Enc}{\sfsty{Enc}}

\DeclarePairedDelimiter\ceil{\lceil}{\rceil}

\DeclareMathOperator{\GL}{GL}
\DeclareMathOperator{\tr}{tr}

\newtheorem{Thm}{Theorem}
\newtheorem{Cor}[Thm]{Corollary}
\newtheorem{Prop}[Thm]{Proposition}
\newtheorem{Lemma}[Thm]{Lemma}
\newtheorem{Def}[Thm]{Definition}
\newtheorem{Rmk}[Thm]{Remark}

\newcommand{\ind}{\text{\color{white}.$\quad$}}

\begin{document}
\title{Orthonormal Sketches for Secure Coded Regression}

\author{
  \IEEEauthorblockN{$\textbf{Neophytos Charalambides}^{\natural}$, $\textbf{Hessam Mahdavifar}^{\natural}$, $\textbf{Mert Pilanci}^{\sharp}$, \textbf{and} $\textbf{Alfred O. Hero III}^{\natural}$}
  \thanks{This work was partially supported by grant ARO W911NF-15-1-0479.$\quad$
  All proofs can be found online in \cite{CMPH22}.}
  \IEEEauthorblockA{$\text{\white.}^{\natural}$EECS Department University of Michigan $\text{\white.}^{\sharp}$EE Department Stanford University\\
  Email: neochara@umich.edu, hessam@umich.edu, pilanci@stanford.edu, hero@umich.edu}
\vspace{-5mm}
}

\maketitle

\begin{abstract}
In this work, we propose a method for speeding up linear regression distributively, while ensuring security. We leverage randomized sketching techniques, and improve straggler resilience in asynchronous systems. Specifically, we apply a random orthonormal matrix and then subsample in \textit{blocks}, to simultaneously secure the information and reduce the dimension of the regression problem. In our setup, the transformation corresponds to an encoded encryption in an \textit{approximate} gradient coding scheme, and the subsampling corresponds to the responses of the non-straggling workers; in a centralized coded computing network. We focus on the special case of the \textit{Subsampled Randomized Hadamard Transform}, which we generalize to block sampling; and discuss how it can be used to secure the data. We illustrate the performance through numerical experiments.
\end{abstract}

\section{Introduction and Preliminaries}
\label{intro}
\vspace{-1mm}

We propose a method to securely speed up linear regression by simultaneously leveraging random projections and distributed computations. Random projections are a classical way of performing dimensionality reduction, and are widely used in algorithmic and learning contexts \cite{Vem05,Woo14,DMMS11,DM16}. Distributed computations in the presence of stragglers have gained a lot of attention in the information theory community. Coding-theoretic approaches have been adopted for this \cite{LLPPR17,reisizadeh2017coded,li2016coded,li2017coding,LSR17,dutta2016short,ramamoorthy2019universally,YSRKSA18,RRG20,CPH20c,CPH20b,OUG20,OBGU20,CMH21}, and fall under the framework of \textit{coded computing} (CC). Data security is also an increasingly important issue in CC \cite{LA20}.

We focus on iterative sketching for steepest descent (SD) in the context of solving overdetermined linear systems. We propose to apply a random orthonormal projection before distributing the data, and then perform stochastic steepest descent (SSD) distributively on the transformed system. A special case of such a projection is the \textit{Subsampled Randomized Hadamard Transform} (SRHT) \cite{DMMS11}, which relates to the \textit{fast Johnson-Lindenstrauss transform} \cite{AC06,JL84}. The benefit of applying an orthonormal matrix transformation is that we rotate and/or reflect the orthonormal basis, which \textit{cannot} be reversed without knowledge of the transformation. This is leveraged to give security guarantees, while simultaneously ensuring that we recover well-approximated gradients, and guaranteeing convergence to the solution of the linear system.

We note that in CC, the workers are assumed to be heterogeneous and all are assumed to have the same expected response time. In the proposed method, we stop receiving computations once a fixed fraction of workers respond, which results in a different sketch at each iteration. A predominant task which has been studied in the CC framework is the gradient computation of differentiable and additively separable objective functions \cite{TLDK17,HASH17,OGU19,CMH20,YA18,RTTD17,CP18,CPE17,WCP19,BWE19,WLS19,KKR19,HYKM19,CHZP18,CPH20a}. These schemes are collectively called \textit{gradient coding} (GC). We note that iterative sketching has proven to be a powerful tool for second-order methods \cite{PW16,LLDP20}, though it has not been explored in first-order methods. Since we consider modified problems at each iteration, the method we propose is an \textit{approximate} GC scheme. Related approaches have been proposed in \cite{CP18,RTTD17,CPE17,WCP19,BWE19,WLS19,KKR19,HYKM19,CHZP18,CPH20a}. Two benefits of our approach are that we do not require a decoding step, and an encoding by the workers; at each iteration.

Another benefit of our proposed approach, is that random projections secure the information from potential eavesdroppers, honest but curious; and colluding workers. We show information theoretic security for the case where a random orthonormal projection is utilized in our sketching algorithm. Furthermore, the security of the SRHT, which is a crucial aspect, has not been extensively studied. Unfortunately, the SRHT is inherently insecure, which we show. We propose a modified projection which guarantees computational security.

There are related works to what we propose. The work of \cite{BP20} focuses on parameter averaging for variance reduction, but only mentions a security guarantee for the Gaussian sketch, derived in \cite{ZWL08}. Another line of work is that of \cite{KSD17,KSDY19}, which focuses on introducing redundancy through equiangular tight frames (ETFs), and partitioning the system into smaller linear systems, and then averaging the solutions of a fraction of them. A drawback is also the fact that some of these ETFs are over $\C$. The authors of \cite{SKD18} study privacy of random projections, though make the assumption that the projections meet the `$\varepsilon$-MI-DP constraint'. Lastly, a secure GC scheme is studied in \cite{YA19}, though this work does not utilize sketching.

The paper is organized as follows. In \ref{coded_LR} we review the framework and background for coded linear regression, and the $\ell_2$-subspace embedding property. In \ref{encr_CC} we present the proposed algorithm, and in \ref{block_SRHT_sec} the special case where the projection is the Hadamard transform, which we refer to as \textit{block-SRHT}. In \ref{security_sec} we present the security guarantee of our algorithm, and the modified version of the block-SRHT; which guarantees computational security. Finally, we present numerical experiments in \ref{exper_sec}; and concluding remarks in \ref{concl_sec}.
\vspace{-4mm}

\section{Coded Linear Regression}
\label{coded_LR}
\vspace{-1mm}

\subsection{Least Squares Approximation and Steepest Descent}
\label{LR_SD}

In linear least squares approximation \cite{DMMS11}, we approximate
\begin{equation}
\label{x_star_pr_lr}
  \xb_{ls}^{\star} = \arg\min_{\xb\in\R^d}\Big\{L_{ls}(\Ab,\bb;\xb)\coloneqq\|\Ab\xb-\bb\|_2^2\Big\}
\end{equation}
where $\Ab\in\R^{N\times d}$ and $\bb\in\R^N$. This corresponds to the regression coefficients $\xb$ of the model $\bb=\Ab\xb+\vec{\varepsilon}$, which is determined by the dataset $\D=\left\{(\ab_i,b_i)\right\}_{i=1}^N\subsetneq \R^d\times\R$ of $N$ samples, where $(\ab_i,b_i)$ represent the features and label of the $i^{th}$ sample, i.e. $\Ab=\big[\ab_1 \ \cdots \ \ab_N \big]^T$ and $\bb=\big[b_1 \ \cdots \ b_N \big]^T.$

We address the overdetermined case where $N\gg d$. Existing exact methods find a solution vector $\xb_{ls}^{\star}$ in $\ow(Nd^2)$ time, where $\xb_{ls}^{\star}=\Ab^{\dagger}\bb$. A common way to approximate $\xb_{ls}^{\star}$ is through SD, which iteratively updates the gradient
$$ g_{ls}^{[t]}\coloneqq\nabla_{\xb}L_{ls}(\Ab,\bb;\xb^{[t]}) = 2\Ab^T(\Ab\xb^{[t]}-\bb) $$
followed by updating the parameter vector: $\xb^{[t+1]}\gets\xb^{[t]}-\xi_t\cdot g_{ls}^{[t]}$. The step-size $\xi_t$ is determined by the central server. The exponent $[t]$ indexes the iteration $t=1,2,3,...$ which we drop when it is clear from the context.

\subsection{The Straggler Problem and Gradient Coding}
\label{str_problem}

Gradient coding is deployed in centralized computation networks, i.e. a central server communicates $\xb^{[t]}$ to $m$ workers; who perform computations and then communicate back their results. The central server distributes the dataset $\D$ among the $m$ workers, to facilitate the solution of optimization problems with additively separable and differentiable objective functions. For linear regression \eqref{x_star_pr_lr}, the data is partitioned as
\vspace{-1mm}
\begin{equation}
\label{part_data}
  \Ab=\Big[\Ab_1^T \ \cdots \ \Ab_K^T\Big]^T \quad \text{ and } \quad \bb=\Big[\bb_1^T \ \cdots \ \bb_K^T\Big]^T
\end{equation}
\vspace{-1mm}
where $\Ab_i\in\R^{\tau\times d}$ and $\bb_i\in\R^{\tau}$ for all $i$, and $\tau=N/K$. To simplify the presentation, we assume that $K|N$. Then
we have $L_{ls}(\Ab,\bb;\xb) = \sum_{i=1}^KL_{ls}(\Ab_i,\bb_i;\xb)$. A regularizer $\mu R(\xb)$ can also be added to $L_{ls}(\Ab,\bb;\xb)$ if desired.

We denote the row vectors of a matrix $\Mb$ by $\Mb_{(i)}$, and the column vectors by $\Mb^{(j)}$. Our embedding results are presented in terms of an arbitrary partition $\N_N=\sqcup_{\iota=1}^K\K_\iota$, for $\N_N\coloneqq\{1,\cdots,N\}$ the index set of $\Mb$'s rows. The notation $\Mb_{(\K_\iota)}$ denotes the submatrix of $\Mb$ comprised of the rows indexed by $\K_\iota$. That is: $\Mb_{(\K_\iota)}=\Ib_{(\K_\iota)}\cdot\Mb$, for $\Ib_{(\K_\iota)}$ the corresponding submatrix of $\Ib_N$. We call $\Mb_{(\K_\iota)}$ the `$\iota^{th}$ block of $\Mb$'.

In GC \cite{TLDK17}, the servers encode their computations $g_i\coloneqq\nabla_{\xb}L_{ls}(\Ab_i,\bb_i;\xb)$; which are then communicated to the central server. We refer to $g_i$'s as \textit{partial gradients}. Once a certain fraction of encoded partial gradients is received, the central server applies a decoding step to recover the gradient $g=\nabla_{\xb}L_{ls}(\Ab_i,\bb_i;\xb)=\sum_{i=1}^K g_i$. This can be computationally prohibitive, and is carried out at every iteration. To the best of our knowledge, the lowest decoding complexity is $\ow\left((s+1)\cdot\ceil{\frac{m}{s+1}}\right)$; where $s$ is the number of stragglers \cite{CMH20}.

In our proposed approach we trade time; by not requiring a decoding step, with accuracy of approximating $\xb_{ls}^{\star}$. Unlike conventional GC schemes, in this paper the workers carry out the computation on the encoded data. The resulting approximation is the solution to the modified least squares problem
\begin{equation}
\label{x_til_pr_lr}
  \xbh_{ls} = \arg\min_{\xb\in\R^d}\Big\{L_{\Sb}(\Ab,\bb;\xb)\coloneqq\|\Sb(\Ab\xb-\bb)\|_2^2\Big\}\
\end{equation}
for $\Sb\in\R^{r\times N}$ a sketching matrix, with $r<N$. This is the core idea in our approximation, where we incorporate iterative sketching with orthonormal matrices; and generalizations of the SRHT for $\Sb$, for our GC approach. The projection, is also what provides security against the workers and eavesdroppers.

\subsection{The $\ell_2$-Subspace Embedding Property}
\label{embedding_subs}

For the analysis of the sketching matrices $\Sbp$ we propose, we consider any orthonormal basis $\Ub\in\R^{N\times d}$ of the column-space of $\Ab$, i.e. $\text{im}(\Ab)=\text{im}(\Ub)$.

Recall that the \textit{$\ell_2$-subspace embedding property} \cite{Woo14} states that any $\yb\in\text{im}(\Ub)$ satisfies:
\begin{align}
\label{subsp_emb_id}
\|\Sbp\yb\|_2\leqslant_\epsilon\|\yb\|_2  \ \iff \ \|\Ib_d-(\Sbp\Ub)^T(\Sbp\Ub)\|_2\leqslant \epsilon
\end{align}
for $\epsilon>0$.\footnote{$\|\vec{a}\|\leqslant_\epsilon\|\vec{b}\| \quad\iff\quad (1-\epsilon)\cdot\|\vec{b}\|\leqslant\|\vec{a}\|\leqslant(1+\epsilon)\cdot\|\vec{b}\|$} In turn, this characterizes the approximation's error of the solution $\xbh_{ls}$ of \eqref{x_til_pr_lr} for $\Sb\gets\Sbp$, as
$$ \|\Ab\xbh_{ls}-\bb\|_2 \leqslant \frac{1+\epsilon}{1-\epsilon}\|\Ab\xb_{ls}^{\star}-\bb\|_2 \leqslant (1+\ow(\epsilon))\|\Ab\xb_{ls}^{\star}-\bb\|_2 $$
and $\|\Ab(\xb_{ls}^{\star}-\xbh_{ls})\|_2\leqslant\epsilon\|(\Ib_N-\Ub\Ub^T)\bb\|_2$.

\section{Block Subsampled Orthonormal Sketches}
\label{encr_CC}

Sampling blocks for sketching least squares has not been explored as extensively as sampling rows, though there has been interest in using ``block-iterative methods'' for solving systems of linear equations \cite{Elf80,Gut06,NT14,RN20}. Our interest in sampling blocks, is to invoke results and techniques from \textit{randomized numerical linear algebra} (RandNLA) to CC. Specifically, we apply the transformation before partitioning the system and sharing it between the workers, who will compute the respective partial gradients. Then, the slowest $s$ workers will be considered as stragglers and disregarded. The proposed sketching matrices are summarised in Algorithm \ref{alg_orthog_sketch}.
\vspace{-2mm}

\begin{algorithm}[h]
\label{alg_orthog_sketch}
\SetAlgoLined
{\small
  \KwIn{$\Ab\in\R^{N\times d}$, $\tau=\frac{N}{K}$, and $q=\frac{r}{\tau}>\frac{d}{\tau}$}
  \KwOut{$\Sbp\in\R^{r\times N}$, and the sketch $\Abph\in\R^{r\times d}$}
  \textbf{Initialize:} $\Omb_{\text{part}}=\bold{0}_{q\times K}$\\
  \textbf{Construct:} $\Pib\in\R^{N\times N}$ {\footnotesize a random orthonormal matrix}\\
  \For{$i=1$ to $q$}
    {
      uniformly sample with replacement $j_i$ from $\N_K$\\
      $(\Omb_{\text{part}})_{i,j_i}=\sqrt{N/r}=\sqrt{K/q}$
    }
  $\Ombp\gets\Omb_{\text{part}}\otimes\Ib_\tau$\\
  $\Sbp\gets\Ombp\cdot\Pib$\\
  $\Abph\gets\Sbp\cdot\Ab$\\
}
\caption{Subsampled Orthonormal Sketches}
\end{algorithm}
\vspace{-2mm}

To construct $\Abph$, we first transform the orthonormal basis $\Ub$ by applying $\Pib$ to $\Ab$. Then, we subsample $q$ many blocks from $\Pib\Ab$, to reduce the dimension. Finally, we normalize by $\sqrt{N/r}$ in order to reduce the variance of the estimator $\Abph$. Analogous steps are carried out on $\Pib\bb$, to construct $\bbh$.

\subsection{Distributed Steepest Descent and Iterative Sketching}
\label{distr_grad_desc}

We now discuss the workers' computational tasks, in the case where SD is carried out distributively. The encoding corresponds to $\Abt=\Gbp\cdot\Ab$ and $\bbt=\Gbp\cdot\bb$ for $\Gbp\coloneqq\sqrt{\frac{N}{r}}\cdot\Pib$, which are then partitioned into $K$ blocks $(\Abt_i,\bbt_i)$; similar to \eqref{part_data}, and distributed to the workers. Specifically, $\Abt_i=\Ib_{(\K_i)}\cdot(\Gbp\Ab)$ and $\bbt_i=\Ib_{(\K_i)}\cdot(\Gbp\bb)$. This differs to most GC schemes, in that the encoding is usually done locally by the workers on the computed results, at each iteration.

If each worker respectively computes $\nabla_{\xb}L_{ls}(\Abt_i,\bbt_i;\xb^{[t]})=2\Abt_i^T(\Abt_i^T\xb^{[t]}-\bbt_i)$ at iteration $t$, and the index multiset of the first $q$ responsive workers is $\Scal^{[t]}$, the aggregated gradient
\vspace{-1mm}
\begin{equation}
\label{gr_update}
  \gh^{[t]} = 2\cdot\sum\limits_{j\in\Scal^{[t]}}\Abt_j^T\left(\Abt_j\xb^{[t]}-\bbt_j\right)
\end{equation}
is equal to the gradient of $L_{\Sb}$ for $\Sb\gets\Sbp^{[t]}$ the induced sketching matrix at each iteration, i.e. $\gh^{[t]}=\nabla_{\xb}L_{\Sbp^{[t]}}(\Ab,\bb;\xb^{[t]})$. The sampling matrix $\Ombp^{[t]}$ and index set $\Scal^{[t]}$, correspond to the $q$ responsive workers. The number of stragglers we mitigate is therefore $s=m-q$. In this way, distributed SD is performed on the modified least squares problem \eqref{x_til_pr_lr}. The number of responsive workers $q$ in the CC framework is determined by the \textit{mother runtime distribution} $F(T)$, where $T$ denotes time \cite{LLPPR17}. Considering homogeneous workers, for a specified stopping time $\tilde{T}$, we have $q=m\cdot F(\tilde{T}\tau/N)$.

In Algorithm \ref{alg_orthog_sketch} and Theorems \ref{subsp_emb_thm_Unif} and \ref{subsp_emb_thm}, we assume sampling uniformly with replacement. In what we just described, we used one replica of each block, thus $m=K$. To compensate for this, more than one replicas of each bock could be distributed. This is not a major concern with uniform sampling, as the probability that the $i^{th}$ block would be sampled more than once is $(q-1)/K^2$, which is negligible for large $K$. Furthermore, we sample \textit{uniformly} at random in Algorithm \ref{alg_orthog_sketch}, as the application of $\Pib$ flattens the \textit{block-leverage scores} \cite{OJXE18,CPH20a}, i.e. they are all approximately equal. That is, for $\Vbt=\Pib\Ub$: $\tilde{\ell}_\iota\coloneqq\|\Vbt_{(\K_\iota)}\|_F^2\approx\frac{d}{K}$ for all $\iota\in\N_K$.

\vspace{-1mm}
\begin{Lemma}
\label{lemma_exp}
At any iteration $t$ of the proposed scheme, with no replications of the blocks across the network, the resulting sketching matrix $\Sb_{[t]}$ satisfies $\E\left[\Sb_{[t]}^T\Sb_{[t]}\right]=\Ib_N$.
\end{Lemma}

By Lemma \ref{lemma_exp}, the Gram matrix of $\Sb_{[t]}$ in expectation satisfies the subspace embedding identity \eqref{subsp_emb_id} with $\epsilon=0$, as $\E\left[\Ub^T\cdot(\Sb_{[t]}^T\Sb_{[t]})\cdot\Ub\right]=\Ub^T\E\left[\Sb_{[t]}^T\Sb_{[t]}\right]\Ub=\Ub^T\Ub=\Ib_d$.

\vspace{-1mm}
\begin{Thm}
\label{GC_SGD_thm}
The proposed GC scheme results in a stochastic steepest descent procedure for
\begin{equation}
\label{mod_pr_Gp}
  \xbh = \arg\min_{\xb\in\R^d}\Big\{L_{\Gbp}(\Ab,\bb;\xb)\coloneqq L_{ls}(\Gbp\Ab,\Gbp\bb;\xb)\Big\} \ .
\end{equation}
Moreover $\E\left[\gh^{[t]}\right]=\frac{q}{K}\cdot g_{ls}^{[t]}$.
\end{Thm}

Note that $\E\left[\gh^{[t]}\right]=\frac{q}{K}\cdot g_{ls}^{[t]}$ means the estimate $\gh^{[t]}$ is unbiased after an appropriate rescaling. This rescaling could be incorporated in the step-size $\xi_t$. The subsampling which takes place; as a consequence of considering the $q$ fastest responses, is the reason the distributive procedure results in a SSD approach for the modified problem.

\vspace{-1mm}
\begin{Lemma}
\label{eq_opt_sols}  
The optimal solution of the modified least squares problem on $L_{\Gbp}$, is equal to the optimal solution $\xb_{ls}^{\star}$ of \eqref{x_star_pr_lr}.
\end{Lemma}

\begin{Cor}
\label{eq_SSD_dor}  
Consider the problems \eqref{x_star_pr_lr} and \eqref{mod_pr_Gp}, which are respectively solved through SD and our iterative sketching scheme. Assume that the two approaches have the same starting point $\xb^{[0]}$ and index set $\Scal^{[t]}$ at each $t$; and $\xit_t=\frac{K}{q}\cdot\xi_t$ the step-sizes used for our scheme. Then, in expectation, our scheme has the same update at each step $t$ as SD at the corresponding update, i.e $\E\left[\xbh^{[t]}\right]=\xb^{[t]}$.
\end{Cor}

By Lemma \ref{eq_opt_sols} and Corollary \ref{eq_SSD_dor}, our iterative sketching scheme approaches the optimal solution of the original problem \eqref{x_star_pr_lr}, by solving the modified regression problem \eqref{mod_pr_Gp}.

\begin{Prop}
\label{contr_rate_thm}
The proposed procedure with a fixed step-size of $\xi$ has a contraction rate of $\gamma_t=\lambda_1(\Bb_{t-1})$ on the error term $\xb^{[t]}-\xb_{ls}^{\star}$, for $\Bb_{t-1}=\big(\Ib_d-2\xi\cdot(\Sbp^{[t-1]}\Ab)^T(\Sbp^{[t-1]}\Ab)\big)$.
\end{Prop}

\begin{Thm}
\label{subsp_emb_thm_Unif}
Fix $\epsilon>0$ such that $\epsilon\ll1/N$. Then, the sketching matrix $\Sbp$ of Algorithm \ref{alg_orthog_sketch} is a $(1\rpm\epsilon)$-embedding of $\Ab$, according to \eqref{subsp_emb_id}. Specifically, for $q=\Theta(\frac{d}{\tau}\log{(2d)}/\epsilon^2)$:
\begin{equation*}
  \Pr\big[\|\Ib_d-\Ub^T\Sbp^T\Sbp \Ub\|_2\leqslant\epsilon\big]\geqslant 1-e^{\Theta(1)}.
\end{equation*}
\end{Thm}

\section{The Block-SRHT}
\label{block_SRHT_sec}

In this section, we focus on a special case of $\Pib$ which can be utilized in Algorithm \ref{alg_orthog_sketch}, the randomized Hadamard transform. The SRHT is comprised of three matrices: $\Omb\in\R^{r\times N}$ a uniform sampling and rescaling matrix of $r$ rows, $\Hbh_N\in\{\rpm1/\sqrt{N}\}^{N\times N}$ the normalized Hadamard matrix for $N=2^n$, and $\Db\in\{0,\rpm1\}^{N\times N}$ with i.i.d. diagonal Rademacher random entries; i.e. it is a signature matrix. The main intuition of the projection is that it expresses the original signal or feature-row in the Walsh-Hadamard basis. Furthermore, $\Hbh_N$ can be applied efficiently due to its structure. In the new basis the block-leverage scores are close to uniform, hence uniform sampling is applied to reduce the effective dimension $N$, whilst the information of the data matrix is maintained.

To exploit the SRHT in distributed GC for linear regression, we generalize it to subsampling \textit{blocks} (i.e. submatrices) instead of rows; of the data matrix, as in Algorithm \ref{alg_orthog_sketch}. We give a subspace embedding guarantee for the block-wise sampling version or SRHT, which characterizes the approximation of our proposed GC for linear regression.

We refer to this special case as the ``block-SRHT'', for which $\Pib$ is taken from the set of orthonormal matrices
\begin{equation}
\vspace{-1mm}
\label{set_Had}  
  H_N\coloneqq\left\{\Hbh_N\cdot\Db\ : \ \Db=\text{diag}(\rpm1)\in\{0,\rpm1\}^{N\times N}\right\},
\end{equation}
where $\Db$ is a random signature matrix with equiprobable entries of +1 and -1, and $\Hbh_N$ for $N=2^n$ is defined by
$$ \Hb_2 = \begin{pmatrix} 1 & 1 \\ 1 & -1 \end{pmatrix} \qquad \qquad \Hbh_N = \frac{1}{\sqrt{N}}\cdot\Hb_2^{\otimes \log_2(N)}\ .$$
The SRHT introduced in \cite{DMMS11} corresponds to the case where we select $\tau=1$, i.e. $K=N$. The main differences in $\Sbp$ is the sampling matrix $\Ombp$, and that $q=r/\tau$ sampling trials take place instead of $r$. Henceforth, we drop the subscript $N$. The limiting computational step in applying $\Sbp$ in \eqref{x_til_pr_lr} is the multiplication by $\Hbh$. The recursive structure of $\Hbh$ permits us to compute $\Sbp\cdot\Ab$ in $\ow(Nd\log N)$ time, by using Fourier based methods. To show that $\Sbp$ constructed based on a $\Pib$ taken from $H_N$ satisfies \eqref{subsp_emb_id}, we first present a key result.

Furthermore, the transformation $\Hbh\Db$ also permits for a very sparse random projection to be applied, instead of $\Ombp$ \cite{AC06}. Also note that the diagonal entries of $\Db$ is the only place in which randomness takes place other than the sampling.

\vspace{-1mm}
\begin{Thm}
\label{subsp_emb_thm}
The block-SRHT $\Sbp$ is a $(1\rpm\epsilon)$-embedding of $\Ab$. In the case that $q=\Theta\big(\frac{d}{\tau}\log(Nd/\delta)\cdot\log(d/\delta)/\epsilon^2\big)$:
\vspace{-1mm}
$$ \Pr\big[\|\Ib_d-\Ub^T\Sbp^T\Sbp \Ub\|_2\leqslant\epsilon\big]\geqslant 1-2\delta \ . $$
\end{Thm}

\vspace{-2mm}
In Subsection \ref{security_sec_SRHT} we alter the transformation $\Hbh\Db$ by permuting its rows. While our $\ell_2$-subspace embedding result remains intact, under mild but necessary assumptions; this transformation now also guarantees computational security.

\vspace{-1mm}
\section{Security of Orthonormal Sketches}
\label{security_sec}

In this section, we discuss the security of the proposed orthonormal-based sketching matrices, and that of the block-SRHT. The main idea behind securing the resulting sketches is that there are infinitely many options of $\Pib$ to select from, making it near-impossible for adversaries to discover the inverse transformation.

To give information-theoretic security-guarantees, we make some mild but necessary assumptions regarding Algorithm \ref{alg_orthog_sketch} and the data matrix $\Ab$. First, we recall the definition of a perfectly secret cryptographic scheme.

\begin{Def}[Ch.2\cite{KL14}]
\label{Sh_secr}
A security scheme $\textup{\textsf{Enc}}$ with message, ciphertext and key spaces $\M$, $\Cc$ and $\K$ respectively is \textbf{Shannon/perfectly secret} w.r.t. a probability distribution $D$ over $\M$, if for all $\bar{m}\in\M$ and all $\bar{c}\in\Cc$:
\vspace{-1mm}
\begin{equation}
  \Pr_{{\substack{m\gets D\\ k\gets\K}}}\left[m=\bar{m}\mid\textup{\textsf{Enc}}_k(m)=\bar{c}\right] = \Pr_{m\gets D}\left[m=\bar{m}\right] \ ,
\end{equation}
\vspace{-2mm}
which is equivalent to the condition that for all $m_0,m_1\in\M$:
\vspace{-2mm}
\begin{equation}
\label{perf_secrecy_id}
  \Pr_{k\gets\K}\left[\textup{\textsf{Enc}}_k(m_0)=\bar{c}\right] = \Pr_{k\gets\K}\left[\textup{\textsf{Enc}}_k(m_1)=\bar{c}\right] \ .
\end{equation}
\end{Def}

For an information-theoretic security-guarantee, $\M$ needs to be finite, which $\M$ in our case corresponds to the set of possible orthonormal bases of the column-space of $\Ab$. This is something we do not have control over, and it depends on the application and distribution from which we assume the data is gathered. Therefore, we assume that $\M$ is finite. For this reason, we consider a finite multiplicative subgroup $(\Otil_\Ab,\cdot)$ of $O_N(\R)$ (thus $\Ib_N\in\Otil_\Ab$, and if $\Qb\in\Otil_\Ab$ then $\Qb^T\in\Otil_\Ab$), which contains all potential orthonormal bases of $\Ab$. Recall that $O_N(\R)$ is a regular submanifold of $\GL_N(\R)$. Hence, we can define a distribution on any subset of $O_N(\R)$.

We then let $\M=\Otil_\Ab$, and assume $\Ub_\Ab$ the $N\times N$ orthonormal basis of $\Ab$ be drawn from $\M$ w.r.t. $D$. We consider $D$ to be the uniform distribution. Furthermore, an inherent limitation of Shannon secrecy is that $|\K|\geq|\M|$.

\vspace{-1mm}
\begin{Thm}
\label{Shan_secr_thm}
In Algorithm \ref{alg_orthog_sketch}, sample $\Pib$ uniformly at random from $\Otil_\Ab$. The application of $\Pib$ to $\Ab$ before partitioning the data, provides Shannon secrecy to $\Ab$ w.r.t. $D$ uniform, for $\K,\M,\Cc$ all equal to $\Otil_\Ab$.
\end{Thm}

\subsection{Securing the SRHT}
\label{security_sec_SRHT}
\vspace{-1mm}

Unfortunately, the guarantee of Theorem \ref{Shan_secr_thm} does not apply to the block-SRHT, as in this case it is restrictive to assume that $\Ub_\Ab\in H_N$. A simple computation on a specific example also shows that this sketching approach does not provide Shannon secrecy. For instance, if $\Ub_0=\Ib_2$, $\Ub_1=\Hbh_2$ and the observed transformed basis $\bar{\Cb}$ has two zero entries, then
\vspace{-1mm}
$$ \Pr_{\Pib\gets H_N}\left[\Pib\cdot\Ub_1=\bar{\Cb}\right] > \Pr_{\Pib\gets H_N}\left[\Pib\cdot\Ub_0=\bar{\Cb}\right]=0. $$
Furthermore, since $\Hbh$ is a known orthonormal matrix, it is a trivial task to invert this projection and reveal $\Db\Ab$. This shows that the inherent security of the SRHT is relatively weak.

\begin{Prop}
  The SRHT does not provide Shannon secrecy.
\end{Prop}

To secure the SRHT and the block-SRHT, we randomly permute the rows of $\Hbh$; before applying it to $\Ab$. That is, for $\Pb\in S_N$ where $S_N\subsetneq\{0,1\}^{N\times N}$ is the permutation group on $N\times N$ matrices, we let $\Hbt\coloneqq\Pb\Hbh\in\{\rpm1/\sqrt{N}\}^{N\times N}$, and the new sketching matrix is
\vspace{-1mm}
\begin{equation}
\label{Sbt_proj}  
  \Sbt = \Ombp\cdot(\Pb\cdot\Hbh)\cdot\Db =  \Ombp\cdot\Hbt\cdot\Db = \Ombp\cdot\Pibt
\end{equation}
for which our flattening result still holds true (Corollary \ref{cor_fl_lem}). The reason we ``garble'' $\Hbh$ is so that the projection applied to $\Ab$ now inherently has more randomness, and allows us to draw from a larger ensemble. Specifically, for a fixed $N$, the block-SRHT has $N^2$ options for the projection of $\Hbh\Db$, while for $\Pibt=\Hbt\Db$ there are $N^2\cdot N!=\ow(N^{1.5+N}e^{-N})$ options for the projection $\Pibt$. Moreover, for
\vspace{-1mm}
\begin{equation}
\label{set_G_Had}  
  \Ht_N\coloneqq\left\{\Pb\cdot\Pib\ : \ \Pb\in S_N \text{ and } \Pib\in H_N\right\}
\end{equation}
the set of all possible \textit{garbled Hadamard transforms}, it follows that $(\Ht_N,\cdot)$ is a finite multiplicative subgroup of $O_N(\R)$. Hence, we can also define a distribution on $\Ht_N$. We also get the benefits of permuting $\Hbh$'s columns without explicitly applying a second permutation, through $\Db$.

By the following Corollary, the result of Theorem \ref{subsp_emb_thm} also holds for the \textit{garbled block-SRHT} (an analogous result is used to prove that the scores of $\Hbh\Db\Ab$ are uniform). Thus, we can apply any $\Pibt$ from $\Ht_N$ in Algorithm \ref{alg_orthog_sketch}, and get a valid sketch.

\begin{Cor}
\label{cor_fl_lem}
  For $\yb\in\R^N$ a fixed (orthonormal) column vector of $\Ub$, and $\Db\in\{0,\rpm1\}^{N\times N}$ with random equi-probable diagonal entries of $\rpm1$, we have:
  \begin{equation}
  \label{flat_lem_id_tilde}
    \Pr\left[\|\Hbt\Db\cdot\yb\|_\infty> C\sqrt{\log(Nd/\delta)/N}\right]\leqslant\frac{\delta}{2d}
  \end{equation}
  for $0<C\leqslant \sqrt{2+\log(16)/\log(Nd/\delta)}$ a constant.
\end{Cor}

Moreover, the flattening result also holds true for random projections $\Rb$ whose entries are rescaled Rademacher random variables, i.e. $\Rb_{ij}=\rpm1/\sqrt{N}$ with equal probability. The advantage of this is that we have a larger set of projections
\begin{equation*}
  \Rt_N\coloneqq\Big\{\Rb\in\{\rpm1/\sqrt{N}\}^{N\times N}:\Pr[\Rb_{ij}=+1/\sqrt{N}]=1/2\Big\}
\end{equation*}
to draw from. This makes it even harder for an adversary to determine which projection was applied. Specifically $|\Rt_N|=2^{N^2}$, which is significantly larger than $|\Ht_N|$. A drawback of applying such a projection is that it is much slower than its Hadamard-based counterpart.

Next, we provide a computationally secure guarantee for the garbled block-SRHT, i.e. when $\Sbp\gets \Ombp\cdot\Pibt$. Our guarantee against computationally bounded adversaries, relies heavily on the assumption that one-way functions (OWFs) exists. Even though OWFs are ``minimal'' cryptographic objects, it is not known whether such functions exist. Proving their existence is non-trivial, as this would then imply that $\textsf{P}\neq\textsf{NP}$. In practice however, this is not unreasonable to assume.
\vspace{-1mm}

\begin{Thm}
\label{SRHT_comp_sec_thm}
Under the assumption that one-way permutations exist, the garbled sketching matrix $\Sbp\gets \Ombp\cdot\Pibt$ is computationally secure against polynomial-bounded adversaries.
\end{Thm}
\vspace{-2mm}

\subsection{Exact Gradient Recovery}
\label{exact_grad_subsec}
\vspace{-1mm}

In the case where the \textit{exact} gradient is desired, one can use the proposed orthonormal projections to encrypt the information from the workers, while requiring that the computations from all the workers are received. From Theorems \ref{Shan_secr_thm} and \ref{SRHT_comp_sec_thm}, we know that under certain assumptions we can secure $\Ab$.

Since the projections are orthonormal, it follows that $\gh^{[t]}=g_{ls}^{[t]}$. Thus, as long as all workers respond, the aggregated gradient is equal to the exact gradient. One can utilize this idea to encrypt other distributive computations, e.g. matrix multiplication, logistic regression. This resembles a homomorphic encryption scheme, but is by no means fully-homomorphic.

\section{Experiments}
\label{exper_sec}

We compared our proposed distributed GC schemes to analogous approaches where the projection $\Pib$ is a Gaussian sketch or a Rademacher random matrix. Our approach was found to outperform both of these sketching methods in terms of convergence and approximation error.

We also compared our approach with uncoded (regular) SD and mini-batch SSD. Random matrices $\Ab\in\R^{2000\times40}$ with non-uniform block-leverage scores were generated for the experiments. Standard Gaussian noise was added to an arbitrary vector from im$(\Ab)$, to define $\bb$. We considered $K=100$ blocks, thus $\tau=20$. Each experiment was carried out six times, and we report the average in our plot. For the experiments in Figure \ref{log_res_err_sparse} we ran a total of 100 iterations, and varied $\xi$ for each experiment by logarithmic factors of the optimal step-size $\xi_{\text{opt}}\coloneqq2/\sigma_{\max}(\Ab)^2$. The effective dimension $N$ was reduced to $r=1000$ in all experiments.

In Figure \ref{log_res_err_sparse} we show how the residual error $\|\xb_{ls}^{\star}-\xbh\|_2$ behaves, with different step-sizes. In the depicted experiment, we considered a sparse matrix $\Ab$. In analogous experiments where we considered a dense matrix, or a matrix drawn from a $t$-distribution, the behaviors were similar. In all cases, the order of the magnitude of the residual error was the same.

\begin{figure}[h]
  \centering
    \includegraphics[scale=.1]{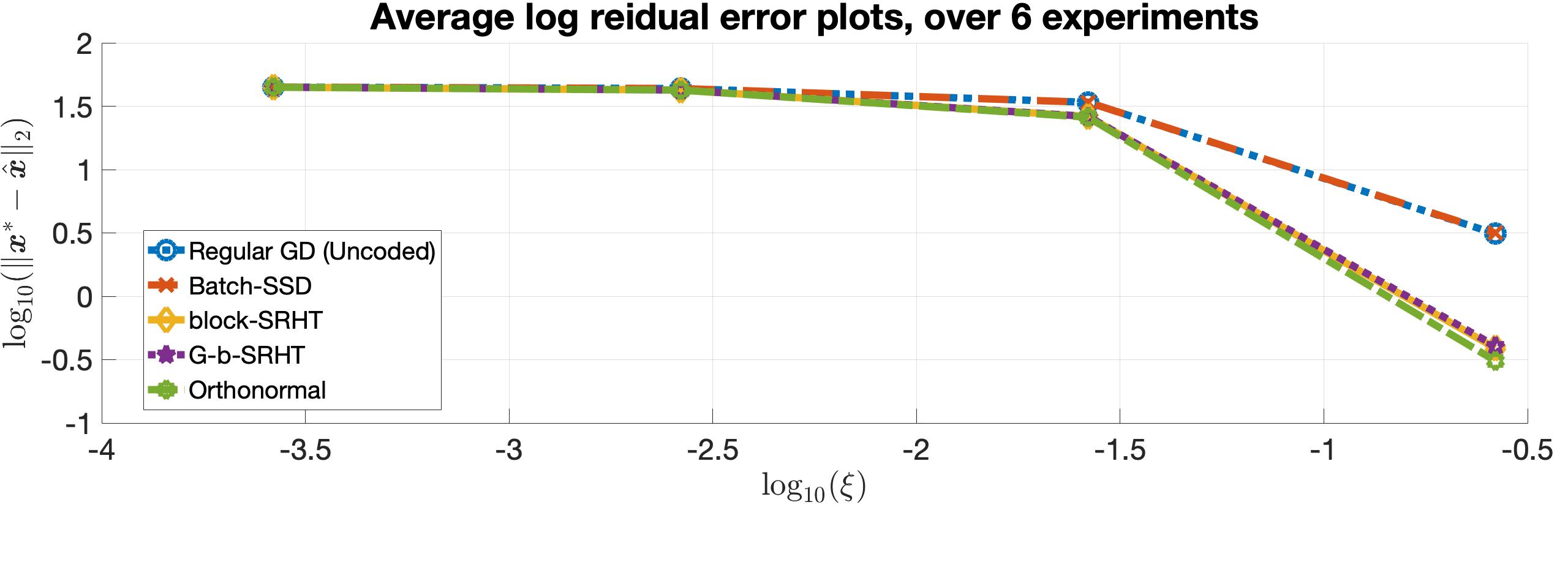}
    \caption{$\log$ residual error, for $\Ab$ sparse.}
  \label{log_res_err_sparse}
\end{figure}

In Figure \ref{conv_per_it} we present the residual error at each iteration, in the case where $\Ab$ was drawn from a $t$-distribution. We considered the case where the step-size was fixed at $\xi=10^2\cdot\xi_{\text{opt}}$. It is evident, that our proposed sketches result in faster convergence of $\xbh$ per iteration, than SD and SSD. This was also the case when a Gaussian projection was applied.

\begin{figure}[h]
  \centering
    \includegraphics[scale=.1]{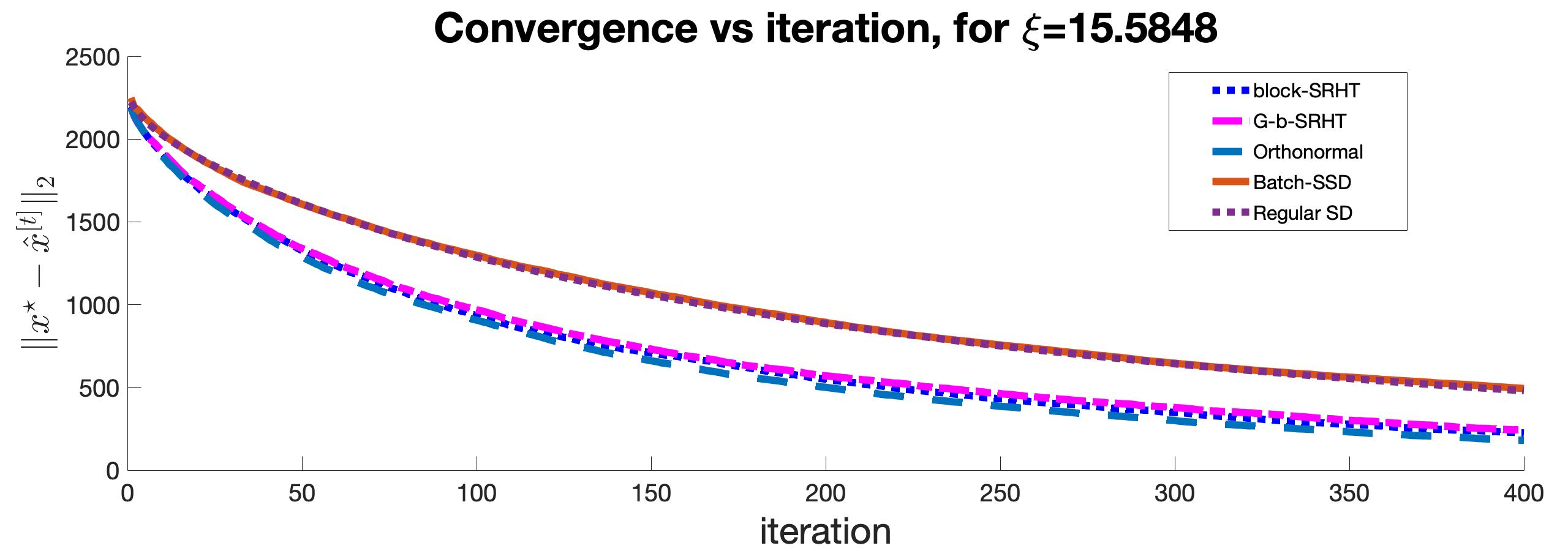}
    \caption{Error at each iteration.}
  \label{conv_per_it}
\end{figure}

Lastly, we show the resulting block-leverage scores after applying the projections, in Figure \ref{flattening_t_distr}. The flattening of these scores is precisely what permitted us to sample uniformly through $\Ombp$, and prove Theorems \ref{subsp_emb_thm_Unif} and \ref{subsp_emb_thm}.

\begin{figure}[h]
  \centering
    \includegraphics[scale=.095]{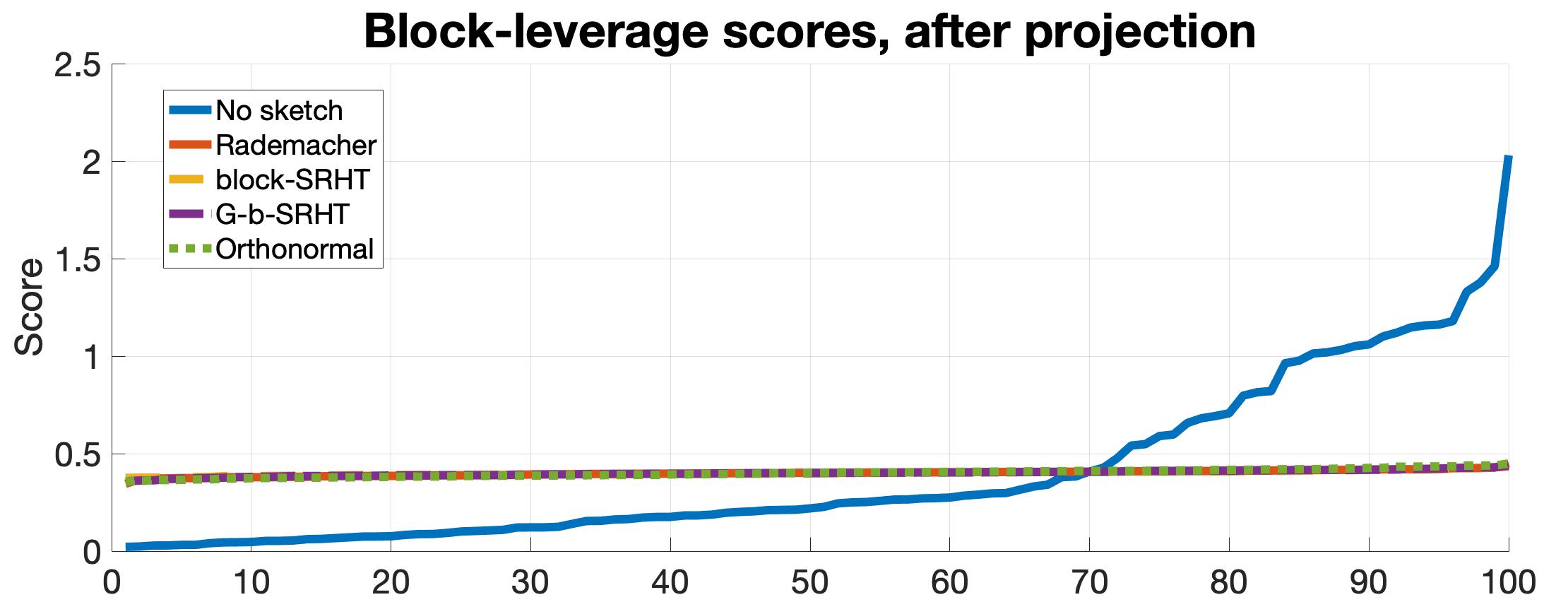}
    \caption{Flattening of block-scores, for $\Ab$ following a $t$-distribution.}
  \label{flattening_t_distr}
\end{figure}

\section{Concluding Remarks}
\label{concl_sec}

In this work, we proposed approximately solving a linear system by distributively leveraging iterative sketching and performing first-order SD simultaneously. In doing so, we benefit from both (approximate) CC and RandNLA. A difference between this and other works is that the resulting sketches are sampling \textit{blocks} uniformly, after applying random orthonormal projections. The benefit is that by considering a large ensemble of orthonormal matrices to pick from, under necessary assumptions, we guarantee information theoretic security while performing the computations. This approach also enables us to not require encoding and decoding steps at every iteration. We also studied the special case where the projection is the Hadamard transform, and discussed its security limitation. To overcome this, we proposed a modified `garbled block-SRHT', which guarantees computational security.

We note that applying orthonormal random matrices also secures coded matrix multiplication. There is a benefit when applying a garbled Hadamard transform in this scenario, as the complexity of multiplication resulting from the sketching is less than that of regular multiplication. Also, if such a random projection is used before performing $CR$-multiplication distributively \cite{CPH20c}, the approximate result will be the same.

Moreover, our dimensionality reduction algorithm can be utilized by a single server, to store a very large data-matrix.


\bibliographystyle{IEEEtran}
\bibliography{refs}

\appendices
\section{Proofs of Section \ref{distr_grad_desc}}

\begin{proof}{[Lemma \ref{lemma_exp}]}
The only difference in $\Sbp^{[t]}$ at each iteration, is $\Scal^{[t]}$ and $\Ombp^{[t]}$. This corresponds to a uniformly random selection of $q$ out of $K$ batches of the data which determine the gradient at iteration $t$ --- all blocks are scaled by the same factor $\sqrt{K/q}$ in $\Ombp^{[t]}$. Let $\Qcal$ be the set of all subsets of $\N_K$ of size $q$. Then
\vspace{-2mm}
\begin{align*}
  \E\big[\Sb_{[t]}^T\Sb_{[t]}\big] &= \sum_{\Scal^{[t]}\in\Qcal}\frac{1}{{K\choose q}}\cdot\left(\Sb_{[t]}\cdot\Sb_{[t]}\right)\\
  &= \frac{1}{{K\choose q}}\sum_{\Scal^{[t]}\in\Qcal}\sum_{i\in\Scal^{[t]}}{\left(\sqrt{K/q}\right)^2}\cdot\Pib_{(\K_i)}^T\Pib_{(\K_i)}\\
  &= \frac{{K-1\choose q-1}}{{K\choose q}}\sum_{i=1}^K\frac{K}{q}\cdot\Pib_{(\K_i)}^T\Pib_{(\K_i)}\\
  &= \frac{{K-1\choose q-1}\cdot\frac{K}{q}}{{K\choose q}}\sum_{i=1}^K\Pib_{(\K_i)}^T\Pib_{(\K_i)}\\
  &= \Pib^T\Pib\\
  &= \Ib_N
\end{align*}
where ${K-1\choose q-1}$ is the number of sets in $\Qcal$ which include $i$, for each $i\in\N_K$.
\end{proof}

\begin{proof}{[Theorem \ref{GC_SGD_thm}]}
The only difference in $\Sbp^{[t]}$ at each iteration, is $\Scal^{[t]}$ and $\Ombp^{[t]}$. This corresponds to a uniformly random selection of $q$ out of $K$ batches of the data which determine the gradient at iteration $t$ --- all blocks are scaled by the same factor $\sqrt{K/q}$ in $\Ombp^{[t]}$. By \eqref{gr_update}, the gradient update is equal to that of a stochastic steepest descent procedure.

We break up the proof of the second statement by first showing that $\E\left[\gh^{[t]}\right]=\gt^{[t]}$; for $\gt$ the gradient in the basis $\Pib\Ub$, and then showing that $\E\left[\gt^{[t]}\right]=\frac{q}{K}\cdot g_{ls}^{[t]}$.

Let $\Qcal$ be the set of all subsets of $\N_K$ of size $q$, $\gh_{\Scal^{[t]}}$ the gradient determined by the index set $\Scal^{[t]}$, and $\gt_i^{[t]}$ the respective partial gradients at iteration $t$. Then
\vspace{-2mm}
\begin{align*}
  \E\left[\gh^{[t]}\right] &= \sum_{\Scal^{[t]}\in\Qcal}\frac{1}{{K\choose q}}\cdot\gh_{\Scal^{[t]}}\\
  &= \frac{1}{{K\choose q}}\sum_{\Scal^{[t]}\in\Qcal}\sum_{i\in\Scal^{[t]}}{\left(\sqrt{K/q}\right)^2}\cdot\gt_i^{[t]}\\
  &= \frac{{K-1\choose q-1}}{{K\choose q}}\sum_{i=1}^K\frac{K}{q}\cdot\gt_i^{[t]}\\
  &= \sum_{i=1}^K\gt_i^{[t]}\\
  &= \gt^{[t]}
\end{align*}
where ${K-1\choose q-1}$ is the number of sets in $\Qcal$ which include $i$, for each $i\in\N_K$.

We denote the resulting partial gradient on the sampled index set $\Scal^{[t]}$ of the gradient on \eqref{x_star_pr_lr} at iteration $t$; i.e. $g_{ls}^{[t]}$, by $g_{\Scal^{[t]}}$, and the individual partial gradients by $g_i^{[t]}$. Using the same notation as above, we get that
\vspace{-2mm}
\begin{align*}
  \E\left[\gt^{[t]}\right] &= \sum_{\Scal^{[t]}\in\Qcal}\frac{1}{{K\choose q}}\cdot g_{\Scal^{[t]}}\\
  &= \frac{1}{{K\choose q}}\sum_{\Scal^{[t]}\in\Qcal}\sum_{i\in\Scal^{[t]}}g_i^{[t]}\\
  &= \frac{{K-1\choose q-1}}{{K\choose q}}\sum_{i=1}^Kg_i^{[t]}\\
  &= \frac{q}{K}\cdot\sum_{i=1}^K\gt_i^{[t]}\\
  &= \frac{q}{K}\cdot g^{[t]}
\end{align*}
which completes the proof.
\end{proof}

\begin{proof}{[Lemma \ref{eq_opt_sols}]}
Since $\Pib$ is an orthonormal matrix, the solution of the least squares problem with the objective $L_{\Gbp}(\Ab,\bb;\xb)$ is equal to the optimal solution \eqref{x_star_pr_lr}, as
\begin{align*}
  \xbh &= \arg\min_{\xb\in\R^d}\|\Gbp(\Ab\xb-\bb)\|_2^2 \\
  &= \arg\min_{\xb\in\R^d}\|\Pib(\Ab\xb-\bb)\|_2^2 \\
  &= \arg\min_{\xb\in\R^d}\|\Ab\xb-\bb\|_2^2\\
  &= \xb_{ls}^{\star}\ .
\end{align*}
\end{proof}

\begin{proof}{[Corollary \ref{eq_SSD_dor}]}
We prove this by induction. From our assumptions we have a fixed starting point $\xb^{[0]}$, for which $\xbh^{[0]}=\xb^{[0]}$. Our base case is therefore $\E[\xbh^{[0]}]=\E[\xb^{[0]}]=\xb^{[0]}$. For the inductive hypothesis, we assume that $\E[\xbh^{[\tau]}]=\xb^{[\tau]}$ for $\tau\in\N$.

It then follows that at step $\tau+1$ we have
\begin{align*}
  \E\big[\xbh^{[\tau+1]}\big] &= \E\big[\xbh^{[\tau]}-\xit_{\tau}\cdot\gh^{[\tau]}\big]\\
  &= \E\big[\xbh^{[\tau]}\big]-\frac{K}{q}\cdot\xi_\tau\cdot\E\big[\gh^{[\tau]}\big]\\
  &= \xb^{[\tau]}-\frac{q}{K}\cdot\left(\frac{K}{q}\cdot\xi_\tau\right)\cdot g_{ls}^{[\tau]}\\
  &= \xb^{[\tau]}-\xi_\tau\cdot g_{ls}^{[\tau]}\\
  &= \xb^{[\tau+1]}
\end{align*}
which completes the inductive step.
\end{proof}

Recall that the contraction rate of an iterative process given by a function $f(x^{[t]})$ is the constant $\gamma$ for which at each iteration we are guaranteed that $f(x^{[t+1]})\leqslant \gamma\cdot f(x^{[t]})$ for $\gamma\in(0,1)$, therefore $f(x^{[t]})\leqslant \gamma^{t}\cdot f(x^{[0]})$. We characterize the convergence of regular steepest descent; and steepest descent through our iterative sketching approach, to the least squares objective \eqref{x_star_pr_lr}, by determining the respective contraction rates on the error term $\xb^{[t]}-\xb_{ls}^{\star}$.

\begin{proof}{[Theorem \ref{contr_rate_thm}]}
Consider a fixed step-size $\xi$, and denote $\Sbp^{[t]}$ by $\Sb_{[t]}$. The steepest descent parameter update through our procedure is then
$$ \xb^{[t+1]}\gets\xb^{[t]}-2\xi\cdot\Ab^T\Sb_{[t]}^T\Sb_{[t]}(\Ab\xb_{[t]}-\bb) $$
where $\Sb_{[t]}$ may change at each iteration. For regular steepest descent, we have $\Sb_{[t]}\gets\Ib_N$. We define the error at iteration $s$ by $e_t\coloneqq\xb^{[t]}-\xb_{ls}^{\star}$, and let $\Bb_t=(\Ib_d-2\xi\cdot\Ab^T\Sb_{[t]}^T\Sb_{[t]}\Ab)$. It follows that
\begin{align}
  e_{t+1} &= \xb^{[t+1]}-\xb^{[t]} \notag\\
  &= \left(\xb^{[t]}-2\xi\cdot\Ab^T\Sb_{[t]}^T\Sb_{[t]}(\Ab\xb^{[t]}-\bb)\right)-\xb_{ls}^{\star} \notag\\
  &= \xb^{[t]}-2\xi\cdot\left(\Ab^T\Sb_{[t]}^T\Sb_{[t]}\Ab\xb^{[t]}+\Ab^T\Sb_{[t]}^T\Sb_{[t]}\bb\right)-\xb_{ls}^{\star} \label{bound_error_sketch}\\
  &= \Bb_t\xb^{[t]}-\left(\xb_{ls}^{\star}-2\xi\cdot\Ab^T\Sb_{[t]}^T\Sb_{[t]}\bb\right)\\
  &= \Bb_t\xb^{[t]}-\left(\xb_{ls}^{\star}-2\xi\cdot\Ab^T\Sb_{[t]}^T\Sb_{[t]}(\Ab\xb_{ls}^{\star})\right) \notag\\
  &= \Bb_t\left(\xb^{[t]}-\xb_{ls}^{\star}\right) \notag
\end{align}
and therefore $e_{t+1} = \Bb_t\left(\xb^{[t]}-\xb_{ls}^{\star}\right) = \Bb_t\cdot e_t$. This gives us the contraction rate
$$ \|e_{t+1}\|_2^2\leqslant\lambda_1(\Bb_t)^2\cdot\|e_{t}\|_2^2 \quad \implies \quad \gamma_{t+1}=\lambda_1(\Bb_t)\ . $$
\end{proof}

The contraction rate of steepest descent is $\gamma_{SD}=\lambda_1(\Bb_{SD})$ for $\Bb_{SD}=(\Ib_d-2\xi\cdot\Ab^T\Ab)$, as here there is no sketching taking place. This is always smaller than the $\gamma_t$ derived in Theorem \ref{contr_rate_thm}, as
\begin{align*}
  \gamma_t &= \|\Ib_d-2\xi\cdot\Ab^T\Ab\|_2\\
  &= \left\|\Ib_d-2\xi\cdot\Ab^T\Ab+2\xi(\Ab^T\Ab-\Ab^T\Sb_{[t]}^T\Sb_{[t]}\Ab)\right\|_2\\
  &\leqslant \overbrace{\|\Ib_d-2\xi\cdot\Ab^T\Ab\|_2}^{\gamma_{SD}}+2\xi\cdot\left\|\Ab^T\Ab-\Ab^T\Sb_{[t]}^T\Sb_{[t]}\Ab\right\|_2\\
  &= \gamma_{SD}+2\xi\cdot\left\|\Ab^T(\Ib_d-\Sb_{[t]}^T\Sb_{[t]})\Ab\right\|_2\\
  &\leqslant \gamma_{SD}+2\xi\cdot\lambda_1(\Ab)^2\cdot\|\Ib_d-\Sb_{[t]}^T\Sb_{[t]}\|_2
\end{align*}
which is expected, since we cannot do better in terms of convergence rate; than steepest descent. By the fact that each $\Sb_{[t]}$ is a $\ell_2$-subspace embeddings, we conclude that with high probability:
\begin{equation*}
\label{contraction_subs_emb}
  \gamma_t \ \leqslant \  \gamma_{SD}+2\xi\cdot\lambda_1(\Ab)^2\cdot\|\Ib_d-\Sb_{[t]}^T\Sb_{[t]}\|_2 \ \leqslant \ \gamma_{SD}+2\xi\epsilon\cdot\lambda_1(\Ab)^2
\end{equation*}
for all iterations.

Next, we provide the proof of Theorem \ref{subsp_emb_thm}. First, we to present the key results regarding the leverage and block-leverage scores of $\Pib\Ab$ (Lemmas \ref{Lemma_exp_lev_i}, \ref{norm_lvg_bd_lemma}).
Throughout this subsection, by $\ell_i$ we denote the $i^{th}$ leverage score of $\Pib\Ab$ for $\Pib$ a random orthonormal matrix, i.e.
\begin{equation}
\label{lvg_sc_expr}
  \ell_i=\|\Ubt_{(i)}\|_2^2=\|\eb_i^T\Ubt\|_2^2=\eb_i^T\Ubt\Ubt^T\eb_i
\end{equation}
where $\Ubt=\Pib\Ub$; for $\Ub$ the reduced left orthonormal matrix of $\Ab$. By $\eb_i$ we denote the $i^{th}$ standard basis vector of $\R^N$.

\begin{Lemma}
\label{Lemma_exp_lev_i}
For each $i\in\N_N$, we have $\E[\ell_i]=\frac{d}{n}$.
\end{Lemma}

\begin{proof}
By \eqref{lvg_sc_expr}, we have
\begin{align*}
  \E[\ell_i] &= \E\left[\ell_i\tr(\eb_i^T\Ubt\Ubt^T\eb_i)\right]\\
  &= \E\left[\ell_i\tr(\eb_i\eb_i^T\cdot\Ubt\Ubt^T)\right]\\
  &= \sum_{j=1}^N\frac{1}{N}\cdot\tr(\eb_i\eb_i^T\cdot\Ubt\Ubt^T)\\
  &= \frac{1}{N}\cdot\tr\left(\sum_{j=1}^N\eb_i\eb_i^T\cdot\Ubt\Ubt^T\right)\\
  &= \frac{1}{N}\cdot\tr\left(\Ib_N\cdot\Ubt\Ubt^T\right)\\
  &= \frac{1}{N}\cdot\tr\left(\Ubt\Ubt^T\right)\\
  &= \frac{d}{N}\ .
\end{align*}
\end{proof}

Let $\ellb_i$ denote the $i^{th}$ normalized leverage score, i.e. $\ellb_i=\frac{\ell_i}{d}$. The $\iota^{th}$ block block-leverage scores of $\Ab$ is denoted by $\ellg_\iota$, i.e.
\begin{equation}
\label{norm_bl_lvg_sc_expr}
  \ellg_\iota=\frac{1}{d}\cdot\|\Ib_{(\K_\iota)}\Ubt\|_F^2=\frac{1}{d}\cdot\Big(\sum_{j\in\K_\iota}\ell_j\Big)=\sum_{j\in\K_\iota}\ellb_j\ .
\end{equation}
To prove Lemma our results, we first recall Hoeffding's inequality.

\begin{Thm}[Hoeffding's Inequality, \cite{Mah16}]
\label{Hoef_in}
Let $\{X_i\}_{i=1}^m$ be independent random variables such that $X_i\in[a_i,b_i]$ for all $i\in\N_m$, and let $X=\sum_{i=1}^mX_i$. Then
\begin{equation*}
\label{Hoeffding_id}
  \Pr\left[\big|X-\E[X]\big|\geqslant t\right] \leqslant 2\cdot\exp\left\{\frac{-2t^2}{\sum_{j=1}^m(a_i-b_i)^2}\right\}.
\end{equation*}
\end{Thm}

\begin{Lemma}
\label{norm_lvg_bd_lemma}
The normalized leverage scores $\{\ellb_i\}_{i=1}^N$ of $\Pib\Ab$ satisfy
$$ \Pr\left[|\ellb_i-1/N|<\epsilon\right] > 1-2\cdot e^{2\epsilon^2/N} $$
for any $\epsilon>0$.
\end{Lemma}

\begin{proof}{[Lemma \ref{norm_lvg_bd_lemma}]}
We know that $\ell_i\in[0,d]$ for each $i\in\N_N$, thus $\ellb_i\in[0,1]$ for each $i$. By Lemma \ref{Lemma_exp_lev_i}, it follows that
$$ \E[\ellb_i]=\E[\ell_i/d]=\frac{1}{d}\cdot\E[\ell_i]=\frac{1}{N}. $$

Now, fix an $\epsilon>0$. By applying Theorem \ref{Hoef_in}, we get
\begin{equation*}
\label{norm_lvg_bd}
  \Pr\left[|\ellb_i-1/N|\geqslant\epsilon\right] \leqslant 2\cdot e^{-2\epsilon^2/N}
\end{equation*}
thus
\begin{equation*}
  \Pr\left[|\ellb_i-1/N|<\epsilon\right] > 1-2\cdot e^{-2\epsilon^2/N}\ .
\end{equation*}
\end{proof}

\begin{Lemma}
\label{bd_block_lvg_Unif}
  For all $\iota\in\N_K$ and $\K_\iota\subsetneq\N_N$ of size $\tau=N/K$
  $$ \Pr\left[\big|\ellg_\iota-1/K\big|<\tau\epsilon\right] = \Pr\left[\ellg_\iota<_{N\epsilon}1/K\right] > 1-2\tau\cdot e^{-2\epsilon^2/N} $$
  for $\epsilon>0$.
\end{Lemma}

\begin{proof}{[Lemma \ref{bd_block_lvg_Unif}]}
By Lemma \eqref{norm_lvg_bd_lemma}, it follows that
\begin{align*}
  \Pr\left[\big|\ellg_\iota-1/K\big|<\tau\epsilon\right] &> \Pr\left[\bigwedge_{j\in\K_\iota}\big\{|\ellb_i-1/N|<\epsilon\big\}\right]\\
  &> \left(1-2\cdot e^{-2\epsilon^2/N}\right)^\tau\\
  &\overset{\Join}{\approx} 1-2\tau\cdot e^{-2\epsilon^2/N}
\end{align*}
where in $\Join$ we applied the binomial approximation.
\end{proof}

The proof of Corollary \ref{subsp_emb_thm_Unif} is a direct consequence of Lemma \ref{bd_block_lvg_Unif} and Theorem \ref{subsp_emb_thm_lvg}. We note that in our statement we make the assumption that $\ell_\iota=1/K$ for all $\iota$, even though this is not the case, as Lemma \eqref{bd_block_lvg_Unif} allows a small deviation. One could generalize Theorem \ref{subsp_emb_thm_lvg} to accommodate sampling according to \textit{approximate} block-leverage scores, e.g. \cite{DMMW12}. This is not studied in our work.

\begin{Thm}
\label{subsp_emb_thm_lvg}
The sketching matrix $\Sbp$ constructed by sampling blocks of $\Ab$ with replacement according to their normalized block-leverage scores $\{\ellg_\iota\}_{\iota=1}^K$ and rescaling each sampled block by $\sqrt{1/(q\ellg_\iota)}$, is a $(1\rpm\epsilon)$-embedding of $\Ab$; as defined in \eqref{subsp_emb_id}. Specifically, for $q=\Theta(\frac{d}{\tau}\log{(2d)}/\epsilon^2)$:
\begin{equation*}
  \Pr\big[\|\Ib_d-\Ub^T\Sbp^T\Sbp \Ub\|_2\leqslant\epsilon\big]\geqslant 1-e^{\Theta(1)}.
\end{equation*}
\end{Thm}

\section{Proofs of Section \ref{block_SRHT_sec}}

In this appendix, we present two lemmas which we use to bound the entries of $\Vbh\coloneqq\Hbh\Db\Ub$, and its \textit{leverage scores} $\ell_i\coloneqq\|\Vbh_{(i)}\|_2^2$, for which $\sum_{i=1}^N\ell_i=d$. Leverage scores induce a sampling distribution which has proven to be useful in linear regression \cite{DMMW12,Woo14,Mah16,Wan15} and GC \cite{CPH20a}. From these lemmas, we deduce that the leverage scores of $\Hbh\Db\Ab$ are close to being uniform, implying that the \textit{block-leverage scores} \cite{CPH20a} are also uniform, which is precisely what Lemma \ref{bd_lvg_Had} states.

Lemma \ref{fl_lem} is a variant of the Flattening Lemma \cite{AC06,Mah16}, a key result to Hadamard based sketching algorithms, which justifies uniform sampling. In the proof, we make use of the Azuma-Hoeffding inequality; a concentration result for the values of martingales that have bounded differences. We also recall a matrix Chernoff bound \cite[Fact 1]{Woo14}, which we apply to prove our subspace embedding guarantees. Finally, we present proofs of Proposition \ref{prop_SRHT_b} and Theorems \ref{GC_SGD_thm}, \ref{subsp_emb_thm}.

\begin{Lemma}[Azuma-Hoeffding Inequality, \cite{Mah16}]
\label{Az_Hoef_in}
For zero mean random variable $Z_i$ (or $Z_0,Z_1,\cdots,Z_m$ a martingale sequence of random variables), bounded above by $|Z_i|\leqslant \beta_i$ for all $i$ with probability 1, we have
$$ \Pr\bigg[\big|\sum_{j=0}^m Z_j\big|>t\bigg] \leqslant 2\exp\left\{\frac{t^2}{2\cdot\big(\sum_{j=0}^m(\beta_j)^2\big)}\right\}. $$
\end{Lemma}

\begin{Thm}[Matrix Chernoff Bound, {\cite[Fact 1]{Woo14}}]
\label{matr_Chern}
  Let $\Xb_1,\cdots,\Xb_q$ be independent copies of a symmetric random matrix $\Xb\in\R^{d\times d}$, with $\E[\Xb]=0, \|\Xb\|_2\leqslant \gamma$, $\|\E[\Xb^T\Xb]\|_2\leqslant \sigma^2$. Let $\Zb=\frac{1}{q}\sum_{i=1}^q\Xb_i$. Then, $\forall\epsilon>0$:
  \begin{equation}
  \label{matr_Chern_expr}
    \Pr\Big[\|\Zb\|_2>\epsilon\Big]\leqslant2d\cdot\exp\left(-\frac{q\epsilon^2}{\sigma^2+\gamma\epsilon/3}\right).
  \end{equation}
\end{Thm}

\begin{Lemma}[Flattening Lemma]
\label{fl_lem}
  For $\yb\in\R^N$ a fixed (orthonormal) column vector of $\Ub$, and $\Db\in\{0,\rpm1\}^{N\times N}$ with random equi-probable diagonal entries of $\rpm1$, we have:
  \begin{equation}
  \label{flat_lem_id}
    \Pr\left[\|\Hbh\Db\cdot\yb\|_\infty> C\sqrt{\log(Nd/\delta)/N}\right]\leqslant\frac{\delta}{2d}
  \end{equation}
  for $0<C\leqslant \sqrt{2+\log(16)/\log(Nd/\delta)}$ a constant.
\end{Lemma}

\begin{proof}{[Lemma \ref{fl_lem}]}
Fix $i$ and define $Z_j=\Hbh_{ij}\Db_{jj}\yb_j$ for each $j\in\N_N$, which are independent random variables. Since $\Db_{jj}=\vec{D}_j$ are i.i.d. entries with zero mean, so are $Z_j$. Furthermore $|Z_j| \leqslant |\Hbh_{ij}|\cdot|\Db_{jj}|\cdot|\yb_j| = \frac{|\yb_j|}{\sqrt{N}}$, and note that
$$ \sum_{j=1}^NZ_j=(\Hbh\Db\yb)_{i}=\sum_{j=1}^N\Hbh_{ij}\Db_{jj}\yb_j=\langle\Hbh_{(i)}\odot \overbrace{\text{diag}(\Db)}^{\vec{D}},\yb\rangle $$
where $\odot$ is the Hadamard product. By Lemma \ref{Az_Hoef_in}
\begin{align}
\label{pr_sum_Zj}
  \Pr\bigg[\Big|\sum_{j=1}^N Z_j\Big|&>\rho\bigg] \leqslant 2\exp\left\{\frac{-\rho^2/2}{\sum_{j=1}^N(\yb_j/\sqrt{N})^2}\right\}\notag\\
  &= 2\exp\left\{\frac{-N\rho^2}{2\cdot\langle\yb,\yb\rangle}\right\} \overset{\flat}{=} 2\cdot e^{-N\rho^2/2}
\end{align}
where $\flat$ follows from the fact that $\yb$ is a column of $\Ub$. By setting $\rho=C\sqrt{\frac{\log(Nd/\delta)}{N}}$, we get
\begin{align*}
  \Pr\left[\Big|\sum_{j=1}^N Z_j\Big|>C\sqrt{\frac{\log(Nd/\delta)}{N}}\right] &\leqslant 2\exp\left\{-\frac{C^2\log(Nd/\delta)}{2}\right\}\\
  &= 2\left(\frac{\delta}{Nd}\right)^{C^2/2} \overset{\natural}{\leqslant} \frac{\delta}{2Nd}
\end{align*}
where $\natural$ follows from the upper bound on $C$. By applying the union bound over all $i\in\N_N$, we attain \eqref{flat_lem_id}.
\end{proof}

\begin{Lemma}
\label{bd_lvg_Had}  
  For all $i\in\N_N$ and $\{\eb_i\}_{i=1}^N$ the standard basis:
  $$ \Pr\left[\sqrt{\ell_i}\leqslant C\sqrt{d\log(Nd/\delta)/N}\right]\geqslant 1-\delta/2 $$
  for $\ell_i=\|\Vbh_{(i)}\|_2^2$ the $i^{th}$ leverage score of $\Vbh=\Hbh\Db\Ub$.
\end{Lemma}

\begin{proof}{[Lemma \ref{bd_lvg_Had}]}
It is straightforward that the columns of $\Vbh$ form an orthonormal basis of $\Ab$, thus Lemma \ref{fl_lem} implies that for $j\in\N_d$
$$ \Pr\left[\|\Vbh\cdot\eb_j\|_\infty> C\sqrt{\log(Nd/\delta)/N}\right] \leqslant \frac{\delta}{2d} \ . $$
By applying the union bound over all entries of $\Vbh^{(j)}=\Vbh\cdot\eb_j$
\begin{equation}
\label{small_entries}
  \Pr\Bigg[\overbrace{|\eb_i^T\cdot\Vbh\cdot\eb_j|}^{|(\Hbh\Db\Ub)_{ij}|}> C\sqrt{\frac{\log(Nd/\delta)}{N}}\Bigg] \leqslant d\cdot\frac{\delta}{2d} = \delta/2\ .
\end{equation}
We manipulate the argument of the above bound to obtain
$$ \|\eb_i^{T}\cdot\Vbh\|_2=\Big(\sum_{j=1}^d(\Hbh\Db\Ub)_{ij}^2\Big)^{1/2} > { C}\sqrt{d\cdot\frac{\log(Nd/\delta)}{N}} \ , $$
which can be viewed as a scaling of the random variable entries of $\Vbh$. The probability of the complementary event is therefore
$$ \Pr\left[\|\eb_i^T\cdot\Vbh\|_2\leqslant C\sqrt{d\log(Nd/\delta)/N}\right]\geqslant 1-\delta/2 $$
and the proof is complete.
\end{proof}

\begin{Rmk}
The complementary probable event of \eqref{small_entries} can be interpreted as `every entry of $\Vbh$ is small in absolute value'.
\end{Rmk}

\begin{Lemma}
\label{bd_block_lvg_Had}
  For all $\iota\in\N_K$ and $\K_\iota\subsetneq\N_N$ of size $\tau=N/K$
  $$ \Pr\left[\tilde{\ell}_\iota\leqslant C^2d\cdot\log(Nd/\delta)/K\right]>1-\tau\delta/2 \ . $$
  for $0<C\leqslant \sqrt{2+\log(16)/\log(Nd/\delta)}$ a constant.
\end{Lemma}

\begin{proof}{[Lemma \ref{bd_block_lvg_Had}]}
For $\alpha\coloneqq C^2d\cdot\log(Nd/\delta)/N$
$$ \Pr\big[\tilde{\ell}_\iota\leqslant\tau\cdot\alpha\big]>\Pr\big[\{\ell_j\leqslant\alpha:\forall j\in\K_\iota\}\big]\overset{\diamondsuit}{>}(1-\delta/2)^{\tau}$$
where $\diamondsuit$ follows from Lemma \ref{bd_lvg_Had}. By the binomial approximation, we have $(1-\delta/2)^{\tau}\approx 1-\tau\delta/2$.
\end{proof}

Define the symmetric matrices
\begin{equation}
\label{def_Xi}
  \Xb_i = \left(\Ib_d-\frac{N}{\tau}\cdot\Vbh_{(\K^i)}^T\Vbh_{(\K^i)}\right) = \left(\Ib_d-K\cdot \Vbh_{(\K^i)}^T\Vbh_{(\K^i)}\right)
\end{equation}
where $\Vbh_{(\K^i)}=\Vbh_{(\K_\iota)}$ is the submatrix of $\Vbh$ corresponding to the $i^{th}$ sampling trial of our algorithm. Let $\Xb$ be the matrix r.v. of which the $\Xb_i$'s are independent copies. Note that the realizations $\Xb_i$ of $\Xb$ correspond to the sampling blocks of the event in \eqref{subsp_emb_id}. To apply Theorem \ref{matr_Chern}, we show that the $\Xb_i$'s have zero mean, and we bound their $\ell_2$-norm and variance. Their $\ell_2$-norms are upper bounded by
\begin{align}
    \|\Xb_i\|_2 &\leqslant \|\Ib_d\|_2+\|\frac{N}{\tau}\cdot \Vbh_{(\K^i)}^T\Vbh_{(\K^i)}\|_2\notag\\
  &= 1+\frac{N}{\tau}\cdot\|\Vbh_{(\K_\iota)}\|_2^2\notag\\
  &\leqslant 1+\frac{N}{\tau}\cdot\max_{\iota\in\N_K}\left\{\|\Ib_{(\K_\iota)}\cdot\Vbh\|_2^2\right\}\notag\\
  &\leqslant 1+\frac{N}{\tau}\cdot\max_{\iota\in\N_K}\left\{\|\Ib_{(\K_\iota)}\cdot\Vbh\|_F^2\right\} \tag*{$\left[\|\Ab\|_2\leqslant\|\Ab\|_F\right]$}\notag\\
  &\overset{\textdollar}{\leqslant} 1+\frac{N}{\tau}\cdot\left(|\K_\iota|\cdot\max_{j\in\N_N}\left\{\|\eb_j^T\cdot\Vbh\|_2^2\right\}\right)\notag\\
  &\leqslant 1+\frac{N}{\tau}\cdot\big(\tau\cdot(C^2\cdot d\log(Nd/\delta)/N)\big) \tag*{[Lemma \ref{fl_lem}]}\notag\\
  &= 1+C^2\cdot d\log(Nd/\delta) \label{bound_gamma}\\
  &=1+N\alpha\notag
\end{align}
for $\alpha=C^2d\cdot\log(Nd/\delta)/N$ where in $\textdollar$ we used the fact that
$$ \|\Ib_{(\K_\iota)}\cdot\Vbh\|_F^2 = \sum_{j\in\K_\iota}\|\eb_j^T\cdot\Vbh\|_2^2 \leqslant |\K_\iota|\cdot\max_{j\in\K_\iota}\left\{\|\eb_j^T\cdot\Vbh\|_2^2\right\} \ . $$
From the above derivation, it follows that
\begin{align*}
  \|\Vbh_{(\K^i)}\|_2^2 &= \|\Vbh_{(\K^i)}^T\Vbh_{(\K^i)}\|_2 \\
  &\leqslant \frac{\tau}{N}\cdot\left(1+C^2\cdot d\log(Nd/\delta)-\|\Ib_d\|_2\right)\\
  &=\tau C^2d/N\cdot\log(Nd/\delta)\\
  &=\tau\alpha 
\end{align*}
for all $\iota\in\N_K$. By setting $\tau=1$, we get an upper bound on the squared $\ell_2$-norm of the rows of $\Vbh$:
\begin{equation}
\label{bound_out_prod_norm}
  \|\Vbh_l\|_2^2 =\|\Vbh_l\Vbh_l^T\|_2 =\|\Vbh_l^T\Vbh_l\|_2 \leqslant \alpha
\end{equation}
where $\Vbh_l=\Vbh_{(l)}$, for all $l\in\N_N$.

Next, we compute $\Eb\coloneqq\E[\Xb^T\Xb+\Ib_d]$ and its eigenvalues. By the definition of $\Xb$ and its realizations:
{\small
\begin{align*}
  \Xb_i^T\Xb_i &= \left(\Ib_d-N/\tau\cdot \Vbh_{(\K^i)}^T\Vbh_{(\K^i)}\right)^T \cdot \left(\Ib_d-N/\tau\cdot \Vbh_{(\K^i)}^T\Vbh_{(\K^i)}\right)\\
  &= \Ib_d-2\cdot\frac{N}{\tau}\cdot \Vbh_{(\K^i)}^T\Vbh_{(\K^i)} + \left(\frac{N}{\tau}\right)^2\cdot \Vbh_{(\K^i)}^T\Vbh_{(\K^i)}\Vbh_{(\K^i)}^T\Vbh_{(\K^i)}
\end{align*}
}
thus $\Eb$ is evaluated as follows:
{\small
\begin{subequations}
\label{E_eval}  
\begin{align*}
  \E[&\Xb^T\Xb+\Ib_d] = 2\Ib_d - 2\cdot\left(N/\tau\right)\cdot\E\left[ \Vbh_{(\K^i)}^T\Vbh_{(\K^i)}\right] \\ &{\white=}+ \left(N/\tau\right)^2\cdot \E\left[\Vbh_{(\K^i)}^T\Vbh_{(\K^i)}\Vbh_{(\K^i)}^T\Vbh_{(\K^i)}\right]\\
  &= 2\Ib_d - 2\cdot\left(N/\tau\right)\cdot\left({\textstyle\sum_{j=1}^K}K^{-1}\cdot \Vbh_{(\K_j)}^T\Vbh_{(\K_j)}\right) \\ &{\white=}+ \left(N/\tau\right)^2\cdot\left({\textstyle\sum_{j=1}^K}K^{-1}\cdot \Vbh_{(\K_j)}^T\left(\Vbh_{(\K_j)}\Vbh_{(\K_j)}^T\right)\Vbh_{(\K_j)}\right)\\
  &= 2\Ib_d - 2\cdot\left({\textstyle\sum_{l=1}^N} \Vbh_l^T\Vbh_l\right) + (N/\tau)\cdot\left({\textstyle\sum_{l=1}^N} \Vbh_l^T\left(\Vbh_l\Vbh_l^T\right)\Vbh_l\right)\\
  &= K\cdot\left({\textstyle\sum_{l=1}^N} \langle \Vbh_l,\Vbh_l\rangle\cdot \Vbh_l^T\Vbh_l\right)
\end{align*}
\end{subequations}
}
where in the last equality we invoked $\sum_{l=1}^N \Vbh_l^T\Vbh_l=\Ib_d$.

In order to bound the variance of the matrix random variable $\Xb$, we bound the largest eigenvalue of $\Eb$; by comparing it to the matrix
$$ \Fb=K\alpha\cdot\left(\sum_{l=1}^N\Vbh_l^T\Vbh_l\right)=K\alpha\cdot\Ib_d $$
whose eigenvalue $K\alpha$ is of algebraic multiplicity $d$. It is clear that $\Eb$ and $\Fb$ are both real and symmetric; thus they admit an eigendecomposition of the form $\Qb\bold{\Lambda}\Qb^T$. Note also that for all $\yb\in\R^d$:
\begin{align}
  \yb^T\Eb\yb &= K\cdot \yb^T\left(\sum_{l=1}^N \Vbh_l^T\left(\Vbh_l\Vbh_l^T\right)\Vbh_l\right)\yb\notag\\
  &\overset{\sharp}{=} K\cdot\sum_{l=1}^N \langle \yb,\Vbh_l\rangle^2\cdot\|\Vbh_l\|_2^2\notag\\
  &\overset{\flat}{\leqslant} K\alpha\cdot\sum_{l=1}^N \langle \yb,\Vbh_l\rangle^2 \label{bound_quadr_E}\\
  &= K\alpha\cdot\sum_{l=1}^N\yb^T\Vbh_l^T\cdot\Vbh_l\yb\notag\\
  &= \yb^T\left(K\alpha\cdot\sum_{l=1}^N\Vbh_l^T\cdot\Vbh_l\right)\yb\notag\\
  &= \yb^T\Fb\yb\notag
\end{align}
where in $\flat$ we invoked \eqref{bound_out_prod_norm}. By $\sharp$ we conclude that $\yb^T\Eb\yb\geqslant0$, thus $\Fb\succeq\Eb\succeq0$.

Let $\wb_i,\zb_i$ be the unit-norm eigenvectors of $\Eb,\Fb$ corresponding to their respective $i^{th}$ largest eigenvalue. Then
$$ \wb_i^T\left(\Qb_\Eb\bold{\Lambda}_\Eb\Qb_\Eb^T\right)\wb_i = \eb_i^T\cdot\bold{\Lambda}_\Eb\cdot\eb_i = \lambda_i \quad {\white\implies} $$
and by \eqref{bound_quadr_E} we bound this as follows:
\begin{align*}
  \lambda_i = \wb_i^T\Eb\wb_i \leqslant K\alpha\cdot\sum_{l=1}^N \langle \wb_i,\Vbh_l\rangle^2\ .
\end{align*}
Since
$$ \wb_1 = \arg\max_{\substack{\vb\in\R^d\\ \|\vb\|_2=1}}\big\{\wb^T\Eb\vb\big\} \ \implies \ \|\Eb\|_2 = \lambda_1=\wb_1^T\Eb\wb_1\ , $$
and $\Fb\succeq\Eb\geqslant0$, it follows that
\begin{align*}
  \|\Eb\|_2&=\wb_1^T\Eb\wb_1 \leqslant \wb_1^T\Fb\wb_1\\
  &\leqslant \arg\max_{\substack{\vb\in\R^d\\ \|\vb\|_2=1}}\big\{\vb^T\Fb\vb\big\} = \|\Fb\|_2 = K\alpha\ .
\end{align*}
In turn, this gives us 
\begin{align}
  \|\E[\Xb^T\Xb]\|_2 &= \|\Eb-\Ib_d\|_2\notag\\
  &\leqslant \|\Eb\|_2+\|\Ib_d\|_2\notag\\
  &\leqslant \|\Fb\|_2+1\notag\\
  &= K\alpha+1\notag\\
  &\leqslant C^2K\frac{d}{N}\log(Nd/\delta)+1 \notag\\
  &= C^2\frac{d}{\tau}\log(Nd/\delta)+1 \label{bound_var}
\end{align}
hence $\|\E[\Xb^T\Xb]\|_2=O\big(\frac{d}{\tau}\log(Nd/\delta)\big)$.

We now have everything we need to apply Theorem \ref{matr_Chern}.

\begin{Prop}
\label{prop_SRHT_b}
The block-SRHT $\Sbp$ guarantees 
$$ \Pr\Big[\|\bold{I}_d-\Ub^T\Sbp^T\Sbp\Ub\|_2>\epsilon\Big] \leqslant 2d\cdot \exp\left\{\frac{-\epsilon^2\cdot q}{\Theta\left(\frac{d}{\tau}\cdot\log(Nd/\delta)\right)}\right\} $$
for any $\epsilon>0$, and $q=r/\tau>d/\tau$.
\end{Prop}

\begin{proof}{[Proposition \ref{prop_SRHT_b}]}
Let $\{\Xb_i\}_{i=1}^q$ as defined in \eqref{def_Xi} denote $q$ block samples. Let $j(i)$ denote the index of the submatrix which was sampled at the $i^{th}$ random trial, i.e. $\K_{j(i)}=\K_{j(i)}^i$. We then get
{\small
\begin{align*}
  \Zb &= \frac{1}{q}\sum_{i=1}^t \Xb_{j(i)}\\
  &= \frac{1}{q}\cdot\sum_{i=1}^q\left(\Ib_d-\frac{N}{\tau}\cdot \Vbh_{(\K_{j(i)})}^T\Vbh_{(\K_{j(i)})}\right)\\
  &= \Ib_d-\sum_{i=1}^q\left(\sqrt{N/r}\cdot \Vbh_{(\K_{j(i)})}\right)^T\cdot\left(\sqrt{N/r}\cdot \Vbh_{(\K_{j(i)})}\right)\\
  &= \Ib_d-\sum_{i=1}^q\left(\sqrt{N/r}\cdot\Ib_{(\K_{j(i)})}\cdot \Vbh\right)^T\cdot\left(\sqrt{N/r}\cdot\Ib_{(\K_{j(i)})}\cdot \Vbh\right)\\
  &= \Ib_d-\left(\Ombp\Hbh\Db\Ub\right)^T\cdot\left(\Ombp\Hbh\Db\Ub\right)\\
  &= \Ib_d - \Ub^T\Sbp^T\Sbp\Ub\ .
\end{align*}
}

We apply Lemma \ref{matr_Chern} by fixing the terms we bounded: \eqref{bound_gamma} $\gamma=C^2d\log(Nd/\delta)+1$, \eqref{bound_var} $\sigma^2=C^2\frac{d}{\tau}\log(Nd/\delta)+1$, and fix $q$ and $\epsilon$. The denominator of the exponent in \eqref{matr_Chern_expr} is then
\begin{align*}
  \big(C^2&d/\tau\cdot\log(Nd/\delta)+1\big)+\big((C^2d\log(Nd/\delta)+1)\cdot\epsilon/3\big) = \\
  &= C^2d/\tau\cdot\log(Nd/\delta)\cdot\big(1+\epsilon\tau/3\big)+(1+\epsilon/3)\\
  &= \Theta\left(\frac{d}{\tau}\log(Nd/\delta)\right) 
\end{align*}
and the proof is complete.
\end{proof}

\begin{proof}{[Theorem \ref{subsp_emb_thm}]}
By substituting $q$ in the bound of Proposition \ref{prop_SRHT_b} and taking the complementary event, we attain the statement.
\end{proof}

\section{Proofs of Section \ref{security_sec}}

\begin{proof}{[Theorem \ref{Shan_secr_thm}]}
Denote the application of $\Pib$ to a matrix $\Mb$ by $\Enc_\Pib(\Mb)=\Pib\Mb$. We will prove secrecy of this scheme, which then implies that a subsampled version of the transformed information is also secure. Let $\Abg=\Enc_\Pib(\Ab)$ and $\bbg=\Enc_\Pib(\bb)$.

The adversaries' goal is to reveal $\Ab$. To prove that $\Enc_\Pib$ is a well-defined security scheme, we need to show that an adversary cannot learn recover $\Ab$; with only knowledge of $(\Abg,\bbg)$.

For a contradiction, assume an adversary is able to recover $\Ab$ after only observing $(\Abg,\bbg)$. This means that it was able to obtain $\Pib^{-1}$, as the only way to recover $\Ab$ from $\Abg$ is by inverting the transformation of $\Pib$: $\Ab=\Pib^{-1}\cdot\Abg$. This contradicts the fact that only $(\Abg,\bbg)$ were observed. Thus, $\Enc_\Pib$ is a well-defined security scheme.

It remains to prove perfect secrecy according to Definition \ref{Sh_secr}. Observe that for any $\bar{\Ub}\in\M$ and $\bar{\Qb}\in\Cc$
\begin{equation} \Pr_{\Pib\gets\K}\left[\Enc_\Pib(\bar{\Ub})=\bar{\Qb}\right] = \Pr_{\Pib\gets\K}\left[\Pib\cdot\bar{\Ub}=\bar{\Qb}\right] = \ind \end{equation}
\begin{equation} \ind = \Pr_{\Pib\gets\K}\left[\Pib=\bar{\Qb}\cdot\bar{\Ub}^{-1}\right] \overset{\sharp}{=} \frac{1}{|\Otil_\Ab|} = \frac{1}{|\K|} \end{equation}
where $\sharp$ follows from the fact that $\bar{\Qb}\cdot\bar{\Ub}^{-1}$ is fixed. Hence, for any $\Ub_0,\Ub_1\in\M$ and $\bar{\Qb}\in\Cc$ we have
$$ \Pr_{\Pib\gets\K}\left[\Enc_\Pib(\Ub_0)=\bar{\Qb}\right] = \frac{1}{|\K|} = \Pr_{\Pib\gets\K}\left[\Enc_\Pib(\Ub_1)=\bar{\Qb}\right] $$
as required by Definition \ref{Sh_secr}. This completes the proof.
\end{proof}

We note that through the SVD of $\Abg$, the adversaries can learn the singular values and right singular vectors of $\Ab$, since
\begin{equation}
  \Abg=(\Pib\cdot\Ub_\Ab)\cdot\Sigb_\Ab\cdot\Vb_\Ab^T=\Ub_\Abg\cdot\Sigb_\Ab\cdot\Vb_\Ab^T \ .
\end{equation}
Recall that the singular values are unique and, for distinct positive singular values, the corresponding left and right singular vectors are also unique up to a sign change of both columns. We assume w.l.o.g. that $\Vb_\Abg=\Vb_\Ab$ and $\Ub_\Abg=\Pib\cdot\Ub_\Ab$.

Geometrically, the encoding $\Enc_\Pib$ changes the orthonormal basis of $\Ub_\Ab$ to $\Ub_\Abg$, by rotating it or reflecting it; when $\text{det}(\Pib)$ is +1 or -1 respectively. Of course, there are infinitely many ways to do so, which is what we are relying the security of this approach on.

Furthermore, unless $\Ub_\Ab$ has some special structure (e.g., triangular, symmetric, etc.), one cannot use an off-the-shelf factorization to reveal $\Ub_\Ab$. Even though a lot can be revealed about $\Ab$, i.e. $\Sigma_\Ab$ and $\Vb_\Ab$, we showed that it is not possible to reveal $\Ub_\Ab$; hence nor $\Ab$, without knowledge of $\Pib$.

\begin{proof}{[Corollary \ref{cor_fl_lem}]}
The proof is identical to that of Lemma \ref{fl_lem}. The only difference is that the random variable entries $\tilde{Z}_j=\Hbt_{ij}\Db_{jj}\yb_j$ for $j\in\N_N$ and the fixed $i$ now differ, though they still meet the same upper bound
$$|\tilde{Z}_j| \leqslant |\Hbt_{ij}|\cdot|\Db_{jj}|\cdot|\yb_j| = \frac{|\yb_j|}{\sqrt{N}}\ . $$
Since \eqref{pr_sum_Zj} holds true, the guarantees implied by flattening lemma also do, thus the sketching properties of the SRHT are maintained.
\end{proof}

\begin{Rmk}
Since the Lemma \ref{fl_lem} and Corollary \ref{cor_fl_lem} give the same result for the block-SRHT and garbled block-SRHT respectively, it follows that Theorem \ref{subsp_emb_thm} also holds for the garbled block-SRHT.
\end{Rmk}

\begin{Def}[Ch.3 \cite{KL14}]
\label{comp_sec}
A security scheme is \textbf{computationally secure} if any probabilistic polynomial-time adversary succeeds in breaking it, with at most negligible probability. By negligible we mean it is asymptotically smaller than any inverse polynomial function.
\end{Def}

\begin{proof}{[Theorem \ref{SRHT_comp_sec_thm}]}
Assume w.l.o.g. that a computationally bounded adversary observes $\Pibt\Ab$, for which
$\Abph=\Sbp\cdot\Ab=\Ombp\cdot(\Pibt\Ab)$ the resulting sketch of Algorithm \ref{alg_orthog_sketch}, for $\Pibt\in\Ht_N$. To invert the transformation of $\Pibt$, the adversary needs knowledge of the components of $\Pibt$, i.e. $\Hbh$ and $\Pb$. Assume for a contradiction that there exists a probabilistic polynomial-time algorithm which, is able to recover $\Ab$ from $\Pibt\Ab$. This means that it has revealed $\Pb$, so that it can compute
$$ \overbrace{(\Db\Hbh\Pb^T)}^{\Pibt^T=\Pibt^{-1}}\cdot(\Pb\Hbh\Db)\cdot\Ab = \Pibt^{-1}\cdot\Pibt\cdot\Ab = \Ab\ , $$
which contradicts the assumption that the permutation $\Pb$ is a OWF. Specifically, recovering $\Ab$ by observing $\Pibt\Ab$ requires finding $\Pb$ in polynomial time.
\end{proof}

Finally, we show that $\gh^{[t]}=g^{[t]}$, which we claimed in Subsection \ref{exact_grad_subsec}. Since $\Pib\in O_N(\R)$ for the suggested projections (except that random Rademacher projection), we have $\Pib^T\Pib=\Ib_N$. It then follows that
\begin{align*}
  \gh^{[t]} &= 2\cdot\sum\limits_{j=1}^K\Abt_j^T\left(\Abt_j\xb^{[t]}-\bbt_j\right)\\
  &= \left(\Pib\Ab\right)^T\cdot\left(\Pib\Ab\xb^{[t]}-\Pib\bb\right)\\
  &= \Ab^T\cdot\left(\Pib^T\Pib\right)\cdot\left(\Ab\xb^{[t]}-\bb\right)\\
  &= g^{[t]}
\end{align*}
and this completes the derivation.

\subsection{Counterexample to Perfect Secrecy of the SRHT}
\label{SRHT_counter_example}

Here, we present an explicit example for the SRHT (which also applies to the block-SRHT), which contradicts Definition \ref{Sh_secr}. Therefore, the SRHT cannot provide perfect secrecy.

Consider the simple case where $N=2$, and assume that $\Hbh_2\in\Otil_\Ab$. Since $(\Otil_\Ab,\cdot)$ is a multiplicative subgroup of $\GL_2(\R)$, we have $\Ib_2\in\Otil_\Ab$. Let $\Ub_0=\Ib_2$ and $\Ub_1=\Hbh_2$.

For $d_1,d_2$ i.i.d. Rademacher random variables and
$$ \Db = \begin{pmatrix} d_1 & 0\\ 0 & d_2 \end{pmatrix} , $$
it follows that
\begin{align*}
  \Cb_0=\left(\Hbh_2\Db\right)\cdot\Ub_0 = \Hbh_2\Db = \frac{1}{2}\begin{pmatrix} d_1 & -d_2\\ d_1 & d_2\end{pmatrix}
\end{align*}
and
\begin{align*}
  \Cb_1=\left(\Hbh_2\Db\right)\cdot\Ub_1 &= \frac{1}{2} \begin{pmatrix} 1 & -1\\ 1 & 1\end{pmatrix} \begin{pmatrix} d_1 & 0\\ 0 & d_2 \end{pmatrix} \begin{pmatrix} 1 & -1\\ 1 & 1\end{pmatrix}\\
  &=\frac{1}{2} \begin{pmatrix} 1 & -1\\ 1 & 1\end{pmatrix} \begin{pmatrix} d_1 & -d_1\\ d_2 & d_2 \end{pmatrix}\\
  &= \frac{1}{2}\begin{pmatrix} d_1-d_2 & -d_1-d_2\\ d_1+d_2 & -d_1+d_2 \end{pmatrix} .
\end{align*}
It is clear that $\Cb_0$ always has precisely two distinct entries, while $\Cb_1$ has three distinct entries; with 0 appearing twice for any pair $d_1,d_2\in\{\rpm1\}$. Therefore, depending on the observed transformed matrix, we can disregard one of $\Ub_0$ and $\Ub_1$ as being a potential choice for $\Pib$.

For instance, if $\bar{\Cb}$ is the observed matrix and it has a two zero entries, then
$$ \Pr_{\Pib\gets H_N}\left[\Pib\cdot\Ub_1=\bar{\Cb}\right] > \Pr_{\Pib\gets H_N}\left[\Pib\cdot\Ub_0=\bar{\Cb}\right]=0 $$
which contradicts \eqref{perf_secrecy_id}.

Note that even if we apply a permutation, as in the case of the garbled block-SRHT, we still get the same conclusion. Hence, the garbled block-SRHT also does not achieve perfect secrecy.

\subsection{Analogy with the One-Time-Pad}

It is worth noting that the encryption resulting by the multiplication with $\Pib$; under the assumptions made in Theorem \ref{Shan_secr_thm}, bares a strong resemblance with the one-time-pad (OTP). This is not surprising, as it is one of the few known perfectly secret encryption schemes.

The main difference between the two, is that the the spaces we work over are the multiplicative group $(\Otil_\Ab,\cdot)$ whose identity is $\Ib_N$ in Theorem \ref{Shan_secr_thm}, and the additive group $\big((\Z/2\Z)^\ell,+\big)$ in the OTP; whose identity is the zero vector of length $\ell$.

As in the OTP, we make the assumption that $\K,\M,\Cc$ are all equal to the group we are working over; $\Otil_\Ab$, which it is closed under multiplication. In the OTP, a message is revealed by applying the key on the ciphertext: if $c=m\oplus k$ for $k$ drawn from $\K$, then $c\oplus k=m$. Analogously here, for $\Pib$ drawn from $\Otil_\Ab$: if $\bar{\Cb}=\Pib\cdot\Ub_\Ab$, then $\bar{\Cb}^T\cdot\Pib=(\Ub_\Ab^T\cdot\Pib^T)\cdot\Pib=\Ub_\Ab^T$. An important difference here is that the multiplication is not commutative.

Also, for two distinct messages $m_0,m_1$ which are encrypted with the same key $k$ to $c_0,c_1$ respectively, it follows that $c_0\oplus c_1=m_1\oplus m_2$ which reveals the the XOR of the two messages. In our case, for the bases $\Ub_0,\Ub_1$ encrypted to $\Cb_0=\Pib\Ub_0$ and $\Cb_0=\Pib\Ub_1$ with the same projection matrix $\Pib$, it follows that $\Cb_0^T\cdot\Cb_1=\Ub_0^T\cdot\Ub_1$.

\end{document}